\documentclass[peerreview]{IEEEtran}
\usepackage[utf8]{inputenc} 
\usepackage[T1]{fontenc}
\usepackage{url}
\usepackage{ifthen}
\usepackage[cmex10]{amsmath}
\usepackage{amsmath,bm}
\usepackage{amsthm}
\usepackage{thmtools}
\usepackage{bbm}
\usepackage{amssymb}
\usepackage{float}
\usepackage{mathtools}
\usepackage{hyperref}
\usepackage{cancel}
\usepackage{graphicx}
\usepackage{xcolor}
\usepackage[english]{babel}
\usepackage{tabularx}
\usepackage{multirow}
\usepackage{stfloats}
\usepackage{cite}
\usepackage[ruled,vlined]{algorithm2e}
\SetKwInput{KwInput}{Input}                
\SetKwInput{KwOutput}{Output}              
\usepackage{algorithmic}
\usepackage{todonotes}
\usepackage{mleftright}\mleftright

\usepackage{graphicx}

\usepackage{tikz}\usetikzlibrary{shapes.multipart}
\usetikzlibrary{positioning}
\tikzset{font={\fontsize{9pt}{12}\selectfont}}
\tikzset{>=latex}
\usetikzlibrary{calc}
\usetikzlibrary{arrows}
\usetikzlibrary{external}
\usetikzlibrary{patterns}
\usetikzlibrary{fit}
\usepackage{changes}

\DeclareMathOperator{\wt}{wt}

\newtheorem{theorem}{Theorem$\!$}
\newtheorem{lemma}{Lemma$\!$}
\newtheorem{claim}{Claim$\!$}
\newtheorem{corollary}{Corollary$\!$}
\newtheorem{proposition}{Proposition$\!$}
\theoremstyle{definition}
\newtheorem{construction}{Construction$\!$}
\newtheorem{definition}{Definition$\!$}
\newtheorem{example}{Example$\!$}

\usepackage{graphicx}

\usepackage{xcolor}

\newcommand{\cB}{\mathcal{B}}
\newcommand{\cC}{\mathcal{C}}

\newcommand{\cM}{\mathcal{M}}

\newcommand{\cR}{\mathcal{R}}

\newcommand{\bS}{\mathbf{S}}

\newcommand{\mybold}[1]{\bm{#1}}
\newcommand{\ba}{{\mybold{a}}}
\newcommand{\bb}{{\mybold{b}}}
\newcommand{\bc}{{\mybold{c}}}
\newcommand{\bd}{{\mybold{d}}}
\newcommand{\be}{{\mybold{e}}}

\newcommand{\bs}{{\mybold{s}}}

\newcommand{\bu}{{\mybold{u}}}
\newcommand{\bv}{{\mybold{v}}}

\newcommand{\bx}{{\mybold{x}}}
\newcommand{\by}{{\mybold{y}}}
\newcommand{\bz}{{\mybold{z}}}

\newcommand{\bepsilon}	{\mybold{\epsilon}}

\newcommand{\bphi}		{\mybold{\phi}}

\newcommand{\bpsi}		{\mybold{\psi}}

\DeclareMathOperator{\VT}{VT}

\newcommand{\defeq}{\mathrel{\stackrel{\makebox[0pt]{\mbox{\normalfont\tiny def}}}{=}}}

\title{Codes for Correcting Asymmetric Adjacent Transpositions and Deletions}
\author{\textbf{Shuche Wang}\IEEEauthorrefmark{1},
        \textbf{Van Khu Vu}\IEEEauthorrefmark{4},
and        \textbf{Vincent Y.~F.~Tan}\IEEEauthorrefmark{2}\IEEEauthorrefmark{3}\IEEEauthorrefmark{1}
        \\[0.5mm]
\IEEEauthorblockA{
\IEEEauthorrefmark{1} \small Institute of Operations Research and Analytics, National University of Singapore, Singapore \\[0.5mm]
\IEEEauthorrefmark{2} \small Department of Mathematics, National University of Singapore, Singapore\\[0.5mm]
\IEEEauthorrefmark{3} \small Department of Electrical and Computer Engineering, National University of Singapore, Singapore\\[0.5mm]
\IEEEauthorrefmark{4} \small Department of Industrial Systems Engineering
and Management, National University of Singapore, Singapore\\[0.5mm]
}
{Emails:\, shuche.wang@u.nus.edu,   isevvk@nus.edu.sg,  vtan@nus.edu.sg }

\thanks{
This paper was presented in part at the 2022 IEEE Information Theory Workshop (ITW) ~\cite{wang2022codes} and 2023 IEEE International Symposium on Information Theory (ISIT) ~\cite{wang2023codes}.}
} 

\begin{document}
\maketitle

\begin{abstract}
Codes in the Damerau--Levenshtein metric have been extensively studied recently owing to their applications in DNA-based data storage. 
In particular, Gabrys, Yaakobi, and Milenkovic (2017) designed a length-$n$ code correcting a single deletion and $s$ adjacent transpositions with at most $(1+2s)\log n$ bits of redundancy. In this work, we consider a new setting where both asymmetric adjacent transpositions (also known as right-shifts or left-shifts) and deletions may occur. We present several constructions of the codes correcting these errors in various cases. 
In particular, we design a code correcting a single deletion, $s^+$ right-shift, and $s^-$ left-shift errors with at most $(1+s)\log (n+s+1)+1$ bits of redundancy where $s=s^{+}+s^{-}$. 
In addition, we investigate codes correcting $t$ $0$-deletions, $s^+$ right-shift, and $s^-$ left-shift errors with both uniquely-decoding and list-decoding algorithms. 
Our main contribution here is the construction of
a list-decodable code with list size $O(n^{\min\{s+1,t\}})$ and with at most $(\max \{t,s+1\}) \log n+O(1)$ bits of redundancy, where $s=s^{+}+s^{-}$.
Finally, we construct both non-systematic and systematic codes for correcting blocks of $0$-deletions with $\ell$-limited-magnitude and $s$ adjacent transpositions.
\end{abstract}

\section{Introduction}
There are four well-known string operations: a deletion (delete a symbol from the string), an insertion (insert a symbol into the string), a substitution (substitute a symbol in the string by a different symbol), and an adjacent transposition (transpose or swap two adjacent symbols).
The Levenshtein (edit) distance between two strings is the smallest number of deletions, insertions, and substitutions differing required to transform one string into another. This distance has a long history and has attracted a lot of research in computer science~\cite{levenshtein1966binary,wagner1974string,brakensiek2020constant,helberg2002multiple,mitzenmacher2009survey,cheng2018deterministic}. Codes in the Levenshtein distance have been investigated extensively recently due to theoretical interests and their applications, including racetrack memories~\cite{chee2018codes,chee2018coding,archer2020foosball} and DNA-based data storage~\cite{sima2020optimalbinary,yazdi2017portable,yazdi2015dna,cai2021correcting}.

The Damerau--Levenshtein distance between two strings is the smallest number of deletions, insertions, substitutions, and adjacent transpositions required to transform one string into another. The distance is a generalization of the well-known Levenshtein distance taking into account adjacent transpositions. Similar to the Levenshtein distance, it is possible to compute the exact Damerau--Levenshtein distance of two strings in {\em polynomial} time \cite{zhao2019string}, but it is not known if we can compute the distance in  {\em linear} time.
The Damerau--Levenshtein distance has been used in a variety of applications in the literature, such as spelling mistake correction~\cite{bard2006spelling}, comparing packet traces~\cite{cai2012touching}, determining genes and predicting gene activity~\cite{majorek2014rnase}.
Recently, Gabrys, Yaaboki, and Milenkovic \cite{gabrys2017codes} proposed to study codes in the Damerau--Levenshtein distance owing to their applications in DNA-based data storage. They provided several constructions of codes correcting both deletions and adjacent transpositions. 

\subsection{Related Works}
The problem of constructing codes for correcting synchronization errors, including deletions and insertions, was first investigated by Levenshtein~\cite{levenshtein1965binary} and Ullman \cite{ullman1966near,ullman1967capabilities}. Sticky deletions/insertions are equivalent to asymmetric deletions/insertions via the Gray mapping~\cite{tallini2010efficient}. 
Due to their extensive applications, such as in flash memories~\cite{dolecek2010repetition,mahdavifar2017asymptotically}, racetrack memories~\cite{chee2018coding}, and DNA data storage systems~\cite{jain2017duplication,kovavcevic2018asymptotically}, codes correcting asymmetric deletions/insertions have garnered significant attention recently. Tallini et al.~\cite{tallini2008new,tallini2010efficient,tallini20111,tallini20131,tallini2022deletions} provided a series of results and code designs for correcting these kinds of errors. In addition, Mahdavifar and Vardy \cite{mahdavifar2017asymptotically} provided some efficient encoding/decoding algorithms for an asymptotically optimal code correcting $t$ sticky-insertions with redundancy $t\log n+o(\log n)$ and thus also for an asymptotically optimal code for correcting $t$ $0$-insertions. 

List decoding was first introduced by Elias~\cite{elias1957list} and well studied in the Hamming metric~\cite{elias1991error,guruswami1998improved,guruswami2007algorithmic}. Recently, several works investigate list-decodable codes in the Levenshtein metric. Wachter-Zeh~\cite{wachter2017list} derived a Johnson-like upper bound on the maximum list size when performing list decoding in the Levenshtein metric. Haeupler, Shahrasbi, and Sudan~\cite{haeupler2017synchronization} presented list-decodable codes for insertions and deletions where the number of insertions and deletions grows with the length of the code. Guruswami and H{\aa}stad~\cite{guruswami2021explicit} proposed explicit list-decodable codes for correcting two deletions with low redundancy. Song, Cai, and Nguyen~\cite{song2022list} provided list-decodable codes for single-deletion and single-substitution with list-size two.

Codes correcting adjacent transposition errors have been extensively investigated as codes for shift errors\cite{nunnelley1990analysis,kuznetsov1993application,shamai1993bounds}. Codes correcting asymmetric shift errors have also been studied recently\cite{kovavcevic2019runlength}. 
We note that codes correcting substitutions, deletions, and their combinations have attracted a lot of interest recently~\cite{smagloy2020single,song2022list}. However, there are only a few  constructions for correcting a {\em combination} of adjacent transposition and other kinds of errors. Kl{\o}ve~\cite{klove1995codes} proposed a class of perfect constant-weight codes capable of correcting a single deletion/insertion or an adjacent transposition.  Gabrys, Yaakobi, and Milenkovic~\cite{gabrys2017codes} presented several codes correcting a combination of deletions and adjacent transpositions. If there is a single adjacent transposition or a single deletion, there exist codes correcting the error with at most $\log n +  O(\log \log n)$ bits of redundancy~\cite{gabrys2021beyond}. The best-known codes correcting a single deletion and at most $s$ adjacent transpositions require $(1+2s)\log n$ bits of redundancy~\cite{gabrys2017codes}. However, these codes are not optimal in general.  


{}

\subsection{Main Contribution}

In this work, we design several new families of codes for correcting different combinations of $t$ $0$-deletions and $s$ adjacent transpositions where $t$ and $s$ are constants. $0$-deletion refers to the deletion of the symbol $0$ and \emph{adjacent transposition} is divided into two types \emph{right-shift of $0$:} $01\rightarrow 10$ and \emph{left-shift of $0$:} $10\rightarrow 01$. Besides, we also investigate codes for correcting $t_{\mathrm{b}}$ blocks of $0$-deletions with $\ell$-limited-magnitude and $s$ adjacent transpositions, where at most $t_{\mathrm{b}}$ blocks of $0$s are deleted with the length of each block is at most $\ell$. All notations indicating the number of errors represents the maximum allowable number of errors unless otherwise specified throughout this paper. For brevity, we omit many \emph{at most} before the number of errors.

Different from codes constructed by weight sequences in Abelian groups~\cite{klove1995codes}, we present a code by modifying the well-known Varshamov--Tenengol'ts code, which is capable of correcting a single adjacent transposition or a single $0$-deletion with $\log n+3$ bits of redundancy. See Theorem~\ref{thm:unique1} for details.

Unlike only considering the single deletion case as in \cite{gabrys2017codes}, we construct a code correcting $t$ $0$-deletions and $s$ adjacent transpositions with at most $(t+2s) \log n+o(\log n)$ bits of redundancy. The constructed code is the first that corrects multiple $0$-deletions and multiple adjacent transpositions. See Theorem~\ref{thm:unique_ts} for details. 

We propose an asymptotically optimal code for correcting a single deletion, $s^+$ right-shifts of $0$ and $s^-$ left-shifts of $0$ with redundancy $(1+s)\log (n+s+1)+1$ bits, where $s=s^{+}+s^{-}$. Compared with the code for correcting a single deletion and $s$ adjacent transpositions with $(1+2s)\log (n+2s+1)$ bits of redundancy in \cite{gabrys2017codes}, if we know these $s$ adjacent transpositions contain $s^{+}$ right-shifts of $0$ and $s^{-}$ left-shifts of $0$, the redundancy of our code can be further reduced to at most $(1+s)\log (n+s+1)+1$ bits, where $s=s^{+}+s^{-}$. See Theorem~\ref{thm:del_rightshift} and Corollary~\ref{cor:general_shift} for details.
    
We also investigate list-decodable codes for correcting at most $t$ $0$-deletions, $s^+$ right-shifts of $0$ and $s^-$ left-shifts of $0$. To the best of our knowledge, our results are the first known list-decodable codes for the asymmetric Damerau--Levenshtein distance. See the proof of Theorem~\ref{thm:list} for the construction.

Further, we construct both non-systematic and systematic codes for correcting $t_{\mathrm{b}}$ blocks of $0$-deletions with $\ell$-limited-magnitude and $s$ adjacent transpositions. See the proof of  Theorem~\ref{thm:bchblocks} for the construction.

\subsection{Organization}
The rest of this paper is organized as follows. Section~\ref{sec:notation} provides notations and preliminaries. Section~\ref{sec:unique_t&s} presents three uniquely-decodable codes for correcting asymmetric deletions and adjacent transpositions.
Section~\ref{sec:list} proposes list-decodable codes for correcting asymmetric deletions and adjacent transpositions with low redundancy.
In Section~\ref{sec:tblocks}, we construct both non-systematic and systematic codes that are capable of correcting $t_{\mathrm{b}}$ blocks of $0$-deletions with $\ell$-limited-magnitude and $s$ adjacent transpositions. Finally, Section~\ref{sec:conclusion} concludes this paper.


\section{Notation and Preliminaries}\label{sec:notation}
We now describe the notations used throughout this paper. Let $\mathbb{N}$ denote the set of natural numbers and $\mathbb{Z}$ denote the set of all integers. Let $\Sigma_q$ be the finite alphabet of size $q$ and the set of all strings of finite length over $\Sigma_q$ is denoted by $\Sigma_q^{*}$, while $\Sigma_q^n$ represents the set of all sequences of length $n$ over $\Sigma_q$. Without loss of generality, we assume $\Sigma_q=\{0,1,\dotsc,q-1\}$. For two integers $i<j$, let $[i:j]$ denote the set $\{i,i+1,i+2, \ldots, j\}$. The size of a binary code $\cC\subseteq \Sigma_2^n$ is denoted as $|\cC|$ and its \emph{redundancy} is defined as $n- \log |\cC|$, where all logarithms without a specified base are assumed to be to the base 2. The code is asymptotically optimal if the ratio of the redundancy to the lower bound specified in Corollary~\ref{cor:lower_redundancy} approaches one as the length $n$ grows large. We adopt the following asymptotic notation: for any two non-negative real sequences $(x_n)_{n\ge 1}$ and $(y_n)_{n\ge 1}$, $x_n \sim y_n$ if $\lim_{n\rightarrow \infty} \frac{x_n}{y_n}=1$ and  $x_n \lesssim y_n$ if $\lim\inf_{n\rightarrow \infty} \frac{y_n}{x_n}\ge 1$.  

We write sequences with bold letters, such as $\bx$ and their elements with plain letters, e.g., $\bx=x_1\dotsm x_n$ for $\bx\in \Sigma_q^n$. The length of the sequence $\bx$ is denoted as $|\bx|$ and $\bx_{[i:j]}$ denotes the substring beginning  at index $i$ and ending at index $j$, inclusive of $i$ and $j$. The weight $\wt(\bx)$ of a sequence $\bx$ represents the number of non-zero symbols in it. A \emph{run} is a maximal substring consisting of identical symbols and $n_{\mathrm{r}}(\bx)$ denotes the number of runs of the sequence $\bx$. For a function that maps a vector to another vector, we also indicate it in bold font, e.g., $\bphi(\cdot)$. The $i$th coordinate of $\bphi(\bx)$ is denoted as $\phi(\bx)_i$. In addition, for a sequence $\bu\in\Sigma_q^n$ and $a<q$, denote $(\bu \bmod a)=(u_1 \bmod a, u_2 \bmod a, \dotsc, u_n \bmod a)$.

For a binary sequence $\bx\in\Sigma_2^n$, we can  write it uniquely as $\bx=0^{u_1}10^{u_2}10^{u_3}\dotsc 10^{u_{w+1}}$, where $w=\wt(\bx)$.
\begin{definition}
Define the function $\bphi:\Sigma_2^n\rightarrow \mathbb{N}^{w+1}$ with $\bphi(\bx)\defeq(u_1,u_2,u_3,\dotsc,u_{w+1})\in\mathbb{N}^{w+1}$, where $\bx=0^{u_1}10^{u_2}10^{u_3}\dotsc 10^{u_{w+1}}$ with $w=\wt(\bx)$. 
\end{definition}

\begin{example}
Let $\bx = 0111010100\in\Sigma_2^{10}$ with $w=5$. Then, $\bphi(\bx) =(1,0,0,1,1,2)$.
\end{example}

\begin{definition}
Define the function $\bpsi:\Sigma_2^n\rightarrow\Sigma_2^n$ such that $\bpsi(\bx)=(x_1,x_1+x_2,\dotsc,x_1+x_2+\dotsm+x_n)$.
\end{definition}


\begin{definition}

The \emph{Lee weight} of an element $x_i\in \Sigma_q$ is defined by
\begin{equation*}
    w_{\mathrm{L}}(x_i)=\begin{cases}
        x_i, &{\text{if}}\; 0\leq x_i\leq q/2 \\
        q-x_i, &{\text{otherwise}}
    \end{cases}
\end{equation*}
For a sequence $\bx\in\Sigma_q^n$, the \textit{Lee weight} of $\bx$ is  
\begin{equation*}
    w_{\mathrm{L}}(\bx)=\sum_{i=1}^n w_{\mathrm{L}}(x_i).
\end{equation*}
Define the \textit{Lee distance} of two sequences $\bx,\bx'\in\Sigma_q^n$ as
\begin{equation*}
    d_{\mathrm{L}}(\bx,\bx')=w_{\mathrm{L}}(\bx-\bx').
\end{equation*}
\end{definition}

\begin{example}
Let $\bx=1405234\in\Sigma_6^7$. Then, $w_{\mathrm{L}}(\bx)=1+2+0+1+2+3+2=11$.
\end{example}
\begin{example}
Let $\bx=1405234\in\Sigma_6^7$ and $\bx'=0305333\in\Sigma_6^7$. Then, $\bx-\bx'=1100501$ and $d_{\mathrm{L}}(\bx,\bx')=w_{\mathrm{L}}(\bx-\bx')=4$.
\end{example}


\begin{definition}
Given a sequence $\bx \in \Sigma_2^n$ and two positive integers $t$ and $s,$ let $\cB_{t,s}(\bx)$ be the set of all binary sequences that can be obtained from $\bx$ with at most $t$ $0$-deletions and at most $s$ adjacent transpositions. The set $\cB_{t,s}(\bx)$ is called the {\em error ball} centered at $\bx$ with   radii $t$ and $s$.
\end{definition}

The code $\cC_{t,s}(n)$ is a uniquely-decodable code for correcting at most $t$ $0$-deletions and at most $s$ adjacent transpositions, if $\cB_{t,s}(\bc_1)\cap \cB_{t,s}(\bc_2)=\emptyset$ for all  $\bc_1,\bc_2\in\cC_{t,s}(n)$, $\bc_1\neq \bc_2$. The code $\cC_{t,s}(n,L)$ is a list-decodable code for correcting at most $t$ $0$-deletions and at most $s$ adjacent transpositions with list size $L$ if for any received sequence $\bx'$, there exist at most $L$ codewords  $\{\bc_1,\dotsc,\bc_{L}\}$ in $\cC_{t,s}(n,L)$ such that $\bx' \in \cB_{t,s}(\bc_i)$, where $i\in[1:L]$.

\begin{example}
Let $\bx=0111010100$, the first and last $0$ bits are deleted and two pairs of ((4th, 5th) and (7th, 8th)) adjacent bits are transposed in $\bx=\bcancel{0}11\underline{1}\underline{0}1\underline{0}\underline{1}0\bcancel{0}$. Then, $\bx'=11011100\in \cB_{2,2}(\bx)$.
\end{example}

\begin{claim}
If a $0$-deletion occurs in $\bx$ and we receive $\bx'$, there is an index $i$ such that $\bphi(\bx)_i-1=\bphi(\bx')_i$.
\end{claim}

\begin{claim}
\label{cla:change_trans}
Assume that an adjacent transposition in $\bx$ involves the $i$th and $(i+1)$st bits. Then, the corresponding changes in $\bphi(\bx)$ are as follows:
\begin{enumerate}
    \item $10\rightarrow 01$: $(\phi(\bx)'_i,\phi(\bx)'_{i+1})=(\phi(\bx)_i+1,\phi(\bx)_{i+1}-1)$.
    \item $01\rightarrow 10$: $(\phi(\bx)'_i,\phi(\bx)'_{i+1})=(\phi(\bx)_i-1,\phi(\bx)_{i+1}+1)$.
\end{enumerate}
\end{claim}
\begin{example}
Assume $\bx=01110{\bf{1}}{\bf{0}}100$, $\bphi(\bx)=(1,0,0,{\bf{1}},{\bf{1}},2)$ and the adjacent transpositions occur at indices of $6$ and $7$ in $\bx$. Then, $\bx'=01110{\bf{0}}{\bf{1}}100$ and $\bphi(\bx')=(1,0,0,{\bf{2}},{\bf{0}},2)$, where $(\phi(\bx')_4,\phi(\bx')_5)=(\phi(\bx)_4+1,\phi(\bx)_5-1)$.
\end{example}

The well-known Varshamov--Tenengol'ts (VT) code will be used in this paper, and we will recall the following lemma. For $\bx\in\Sigma_2^n$, we define the VT syndrome of a sequence $\bx$ as $\VT(\bx)=\sum_{i=1}^n ix_i$.

\begin{lemma}[Varshamov--Tenengol'ts (VT) code \cite{levenshtein1966binary}]
For integers $n$ and $a \in [0:n]$, 
\begin{equation*}
    \VT_a(n)=\left\{ \bx\in\Sigma_2^n:\VT(\bx)\equiv a\bmod (n+1) \right\}
\end{equation*}
is capable of correcting a single deletion.
\end{lemma}

Define $\cM_{t,s}(n)$ as the maximal size of binary codes for correcting $t$ $0$-deletions and $s$ adjacent transpositions. Lemma~\ref{pro:upperbound} provides an asymptotic upper bound on $\cM_{t,s}(n)$. A detailed proof is included in Appendix~\ref{app:proof_uppersize}.
\begin{lemma}\label{pro:upperbound}
For $n$ large enough, $\cM_{t,s}(n) \lesssim \frac{2^n}{n^{s+t}}\cdot s!\cdot t!\cdot2^{s+2t}$.
\end{lemma}

From Lemma~\ref{pro:upperbound}, we can obtain a lower bound on the minimal redundancy of any of the code for correcting $t$ $0$-deletions and $s$ adjacent transpositions.

\begin{corollary}\label{cor:lower_redundancy}
The minimal redundancy of any binary code for correcting $t$ $0$-deletions and $s$ adjacent transpositions is lower bounded by $(t+s)\log n-O(1)$.
\end{corollary}


\section{Uniquely-Decodable Codes for asymmetric deletions and adjacent transpositions}\label{sec:unique_t&s}
In this section, we present three uniquely-decodable codes for correcting asymmetric deletions and adjacent transpositions.

\subsection{Codes for correcting a single $0$-deletion or a single adjacent transposition}
In this subsection, we present the first construction of a code correcting a single $0$-deletion or a single adjacent transposition.

\begin{construction}\label{const.1or1}
Let $p$ be a prime such that $p>4n$ and $a\in[0:p-1]$. The code $\cC_{1}(n,a;p)$ is defined as the set of all $\bx\in\Sigma_2^n$ such that the syndrome
\begin{equation*}
    \bS(\bx)=\sum_{i=1}^{w+1}i^2\phi(\bx)_i\equiv a\bmod p
\end{equation*}
where $w=\wt(\bx)$.
\end{construction}
\begin{theorem}\label{thm:unique1}
The code $\cC_{1}(n,a;p)$ in Construction \ref{const.1or1} can correct a single $0$-deletion or a single adjacent transposition with the redundancy of at most $\log n+3$ bits.
\end{theorem}
\begin{proof}
Let $\bx =x_1\dotsm x_n \in \Sigma_2^n$ be the original vector and $\bx'$ be the received vector with a single $0$-deletion or a single adjacent transposition. 

If $\bx' \in \Sigma_2^{n-1}$, then there is a single $0$-deletion. In this case, we can compute the vector $\bphi(\bx')$ and $a'<p$ such that $a'= \bS(\bx') \bmod p$. We note that $d_{\mathrm{L}}(\bphi(\bx),\bphi(\bx'))=1$ and there is an index $i$ such that  $\phi(\bx)_i-1=\phi(\bx')_i$. Hence, $\bS(\bx)-\bS(\bx')=i^2$. That is, $a-a'=i^2 \bmod p$ with $a=\bS(\bx)\bmod p$. Since $i^2 -j^2 \neq 0 \bmod p$ for all $i\neq j$, $i,j<n<p/4$, there exists a unique index $i$ such that $a-a'=i^2 \bmod p.$ Thus, the error can be located and corrected. 

If $\bx' \in \Sigma_2^{n}$, then there is no $0$-deletion and at most a single adjacent transposition. Similar to the previous case, we also compute the vector $\bphi(\bx')$ and $a'<p$ such that $a'= \bS(\bx') \bmod p$. If an adjacent transposition occurs, it is one of the following two types of errors: $10\rightarrow 01$ and $01\rightarrow 10$.

\begin{itemize}
    \item $10\rightarrow 01$: There exists $0\leq i \leq n-1$ such that $a-a'\equiv 2i+1\bmod p$, where $a-a'\in[1:2n-1]$.
    \item $01\rightarrow 10$: There exists $0\leq i \leq n-1$ such that $a-a'\equiv-2i-1\bmod p$, where $a-a'\in[p-2n+1:p-1]$.
\end{itemize}


Since $p>4n$, we can observe that $[1:2n-1]\cap [p-2n+1:p-1]=\emptyset$. Consequently, the two cases mentioned above can be distinguished by comparing $a-a'$ with $2n-1$. Additionally, let another true value of $j\neq i$ satisfy the equation $a-a'\equiv 2j+1 \bmod p$. In this case, we would have $2(i-j)\equiv 0\bmod p$. However, this is impossible as $i,j<n<p/4$. Then, we can accurately determine the type of error and the unique $i$ for which $a-a'=2i+1\bmod p$ or $a-a'=-2i-1\bmod p$. Therefore, whenever either a $0$-deletion occurs or an adjacent transposition occurs, we can always correct the error and recover the original vector.

It is known that for large enough integer $n$, there exists a prime $p$ such that $4n< p< 4n+o(n)$~\cite{baker2001difference}. Hence, by the pigeonhole principle, there exists a code $\cC_{1}(n,a;p)$ of size at least $2^n/(4n+o(n))$. That is, it is possible to construct the code $\cC_{1}(n,a;p)$ with the redundancy of at most $\log n+3$ bits.  Therefore, we conclude that we can correct a single $0$-deletion or a single adjacent transposition with the redundancy of at most $\log n+3$ bits. 
redundancy by at most $3$ bits, independent of $n$. 
\end{proof}

\subsection{Codes for correcting $t$ $0$-deletions and $s$ adjacent transpositions}
In this subsection, we explore the general case where the asymmetric Damerau--Levenshtein distance is adopted. We investigate a code correcting at most $t$ $0$-deletions and $s$ adjacent transpositions, where $t$ and $s$ are constants.

It can be seen that the asymmetric Damerau--Levenshtein distance between two vectors $\bx$ and $\by$ is closely related to the Lee distance between $\bphi(\bx)$ and $\bphi(\by)$. 

\begin{lemma}[Mahdavifar and Vardy \cite{mahdavifar2017asymptotically}]\label{lem:lee_bch}
Let $p$ be a prime such that $p>2(r+1)$. The code $\cC(n,r,p)$ satisfies that for any two codewords $\bc_1,\bc_2 \in \cC(n,r,p)$, $d_{\mathrm{L}}(\bphi(\bc_1),\bphi(\bc_2))\geq 2(r+1)$, which is capable of correcting $r$ $0$-insertions with at most $r\log n +o(\log n)$ bits of redundancy.
\end{lemma}
Then, we use the code $\cC(n,r,p)$ to correct $t$ $0$-deletions and $s$ adjacent transpositions. 

\begin{theorem}\label{thm:unique_ts}
If $r=t+2s$, then the code $\cC(n,r,p)$ can correct at most $t$ $0$-deletions and $s$ adjacent transpositions with redundancy $(t+2s)\log n+o(\log n)$ bits.
\end{theorem}
\begin{proof}
Let $\bx =x_1\dotsc x_n \in \Sigma_2^n$ be the original vector and $\bx'\in \Sigma_2^{n-t}$ be the received vector after $t$ $0$-deletions and $s$ adjacent transpositions. 
Now, we consider two vectors $\bphi(\bx)$ and $\bphi(\bx')$. If an adjacent transposition occurs in $\bx$ and we obtain $\by$, the Lee distance between $\bphi(\bx)$ and $\bphi(\by)$ is $d_{\mathrm{L}}(\bphi(\bx),\bphi(\by))=2$ based on Claim~\ref{cla:change_trans}. Also, if a $0$-deletion occurs in $\bx$ and we obtain $\bz$, the Lee distance between $\bphi(\bx)$ and $\bphi(\bz)$ is $d_{\mathrm{L}}(\bphi(\bx),\bphi(\bz))=1$. 
Hence, if there are at most $s$ adjacent transpositions and $t$ $0$-deletions, the Lee distance between two vectors $\bphi(\bx)$ and $\bphi(\bx')$ is $d_{\mathrm{L}}(\bphi(\bx),\bphi(\bx'))\leq t+2s$. 
Therefore, if we set $r=t+2s$, then the code $\cC(n,r,p)$ with $r=t+2s$ can correct at most $t$ $0$-deletions and $s$ adjacent transpositions with redundancy $(t+2s)\log n+o(\log n)$.
\end{proof}

\subsection{Codes for correcting a single deletion and multiple right-shifts}
In the previous two subsections, we focused on $0$-deletions and arbitrary adjacent transpositions (both $01\rightarrow 10$ and $10\rightarrow 01$ may occur) in the asymmetric Damerau-Levenshtein distance. In this subsection, we propose a code for correcting a single deletion and $s$ right-shifts of $0$ and we show that the redundancy is at most $(1+s)\log(n+s+1)+1$ bits. Notice that the redundancy of our construction differs from the minimum redundancy as stated in Corollary~\ref{cor:lower_redundancy} by a constant. We denote the \emph{adjacent transposition} as $01\rightarrow 10$ or $10\rightarrow 01$, \emph{right-shift of $0$} as $01\rightarrow 10$ and \emph{left-shift of $0$} as $10\rightarrow 01$ throughout this subsection.

\begin{construction}\label{con:del_rightshift}
  The code $\cC_{s}(n,a,b)$ is defined as follows.
  \begin{equation*}
      \cC_{s}(n,a,b)=\{ \bx\in\Sigma_2^n: \VT(\bx)\equiv a \bmod (n+s+1),\; 
     \sum_{i=1}^n x_i \equiv b \bmod 2,\; 
 \bpsi(\bx)\in \cC_{H}(n,2s+1)\},
  \end{equation*}
where $\cC_{H}(n,2s+1)$ is a binary narrow-sense BCH code with minimum Hamming distance $2s+1$~\cite{roth2006introduction}. 
\end{construction}

\begin{proposition}[Gabrys,~Yaakobi,~and Milenkovic \cite{gabrys2017codes}]\label{pro:shift_sub}
If a single adjacent transposition ($01\rightarrow 10$ or $10\rightarrow 01$) occurs in $\bx$, then a single substitution occurs in $\bpsi(\bx)$.
\end{proposition}

\begin{proposition}\label{pro:shift_vtchange}
Assume that $\bx'$ is obtained by applying $s$ right-shifts of $0$ on $\bx$. Then, we have $\VT(\bx)-\VT(\bx')=s$.
\end{proposition}
\begin{proof}
Assume a right-shift of $0$ ($01\rightarrow 10$) occurs at a bit $1$ with index $i$ in $\bx$. The index of this $1$ in $\bx'$ will be $i-1$. Thus, for a single right-shift of $0$, the change of the VT syndrome will be 1. If there are $s$ right-shifts of $0$ occurring in $\bx$, we have $\VT(\bx)-\VT(\bx')=s$.
\end{proof}

\begin{lemma}\label{lem:misinsert}
    The following statements hold:
    \begin{itemize}
        \item Assume a $0$ is deleted before the $p$-th $1$ in $\bx$, and a $0$ is inserted before the $(p+v)$-th $1$ to obtain $\hat{\bx}$. Then, $\bx$ is obtained from $\hat{\bx}$ by at most $v$ adjacent transpositions.
        \item Assume a $1$ is deleted after the $p$-th $0$ in $\bx$, and a $1$ is inserted after the $(p-v)$-th $0$ to obtain $\hat{\bx}$. Then, $\bx$ is obtained from $\hat{\bx}$ by at most $v$ adjacent transpositions.
    \end{itemize}
\end{lemma}
\begin{proof}
Denote the indices of $p$-th $1$, $(p+1)$-th $1$, $\dotsc$, $(p+v-1)$-th $1$ in $\bx$ as $i_p,i_{p+1},\dotsc,i_{p+v-1}$. Then, we notice that the indices of these $1$s in $\hat{\bx}$ are $i_p-1,i_{p+1}-1,\dotsc,i_{p+v-1}-1$. Since $0$ is inserted before $(p+v)$-th $1$, we can swap the $(i_{p+v-1}-1)$-th and $i_{p+v-1}$-th bits and hence $\hat{\bx}_{[i_{p+v-1}:i_{p+v}]}=\bx_{[i_{p+v-1}:i_{p+v}]}$. Continuing this process, $\bx$ can be recovered from $\hat{\bx}$ by at most $v$ adjacent transpositions. The case of deleting $1$ is the same as deleting $0$, hence we have the above two statements.
\end{proof}

\begin{theorem}\label{thm:del_rightshift}
For integers $a\in[0:n+s]$ and $b\in[0:1]$, the code $\cC_{s}(n,a,b)$ can correct a single deletion and $s$ right-shifts of $0$ with at most $(1+s)\log (n+s+1)+1$ bits of redundancy.
\end{theorem}
\begin{proof}
    Given a sequence $\bx\in\Sigma_2^n$, denote the received sequence as $\bx'$ with a single deletion and at most $s$ right-shifts of $0$. We first use the VT syndrome to correct the deletion and then use the decoder of $\cC_{H}(n,2s+1)$ on $\bpsi(\bx)$ to correct the right-shifts of $0$.
    
    Let $\Delta=\VT(\bx)-\VT(\bx')$, $w$ be the weight of $\bx'$, and $p$ be the index of deletion. Then, let $L_0$ be the number of $0$s on the left of the deleted bits in $\bx'$ and $R_0$ on its right. Similarly, denote $L_1$ and $R_1$ as the number of $1$s on the left of the deleted bits in $\bx'$ and $R_1$ on its right, respectively. We have the following cases when recovering $\bx$ from $\bx'$:
    \begin{itemize}
        \item If $\bx'\in\Sigma_2^n$, it means no deletion occurs in $\bx$ and there are at most $s$ right-shifts of $0$. Based on Proposition~\ref{pro:shift_sub}, there are at most $s$ substitutions in $\bpsi(\bx)$. Hence we can recover $\bpsi(\bx)$ by $\bpsi(\bx')$ since $\bpsi(\bx)\in\cC_{H}(n,2s+1)$, and then we can recover $\bx$.
        
        \item If $\bx'\in\Sigma_2^{n-1}$ and assume a $0$ is deleted. From Proposition~\ref{pro:shift_vtchange}, we have $\Delta = R_1 + k$, where $k$ represents the actual number of right-shifts of $0$s. To begin, we can recover $\hat{\bx}$ by inserting $0$ at the rightmost index of $\max(\Delta - s, 0)$ $1$s, denoting the insertion index as $p'$. Considering the knowledge that there are at most $s$ right-shifts of $0$s and $\Delta = R_1 + k$, we can deduce that $\max(R_1 + k - s, 0) \leq R_1$. In the case where $R_1 + k - s > 0$, there are $(s - k)$ 1s between index $p$ and $p'$.
        On the other hand, when $R_1 + k - s \leq 0$, there are $R_1$ 1s between index $p$ and $p'$, where $R_1 \leq s - k$.
        Utilizing Case 1 of Lemma~\ref{lem:misinsert}, we find that there are at most $(s - k)$ adjacent transpositions between $\hat{\bx}$ and $\bx$. Additionally, $\bx$ experiences $k$ right-shifts of $0$s. Therefore, $\bx$ is derived from $\hat{\bx}$ through a total of $s$ adjacent transpositions. Consequently, we can recover $\bpsi(\bx)$ from $\bpsi(\hat{\bx})$ and subsequently obtain $\bx$.

        \item If $\bx'\in\Sigma_2^{n-1}$ and we assume that a $1$ is deleted, we can apply Proposition~\ref{pro:shift_vtchange}. As a result, $\Delta = p + R_1 + k = w + L_0 + k + 1$. To recover $\hat{\bx}$, we insert a $1$ at the leftmost index of the $\max(\Delta - w - s - 1, 0)$ $0$s. Similar to Case 2, since $\Delta = w + L_0 + k + 1$ and we insert a $1$ at the leftmost index of the $\max(L_0 + k - s, 0)$ $0$s. Based on Case 2 of Lemma~\ref{lem:misinsert}, we can deduce that there are at most $(s - k)$ adjacent transpositions between $\hat{\bx}$ and $\bx$. Likewise, $\bx$ is derived from $\hat{\bx}$ through a total of $s$ adjacent transpositions. Consequently, we can recover $\bpsi(\bx)$ from $\bpsi(\hat{\bx})$ and subsequently obtain $\bx$.
    \end{itemize}
    
    It is worth noticing that the length of the retrieved sequence $\bx'$ can help distinguish whether a deletion has occurred. Case 2 and Case 3 can be distinguished based on the constraint of $\sum_{i=1}^{n} x_i\equiv b\bmod 2$, from which we know whether the deleted bit is $0$ or $1$. Hence, the code $\cC_{s}(n,a,b)$ is capable of correcting a single deletion and $s$ right-shifts of $0$.
    
    There are three constraints on the sequence $\bx\in\cC_{s}(n,a,b)$ including the set of constraints that define the VT code, a parity check bit, and a linear binary $(n,2s+1)$-code. It can be easily verified that the redundancy of the code $\cC_{s}(n,a,b)$ is $\log(n+s+1)+s\log n+1$. Thus, the redundancy of the code $\cC_{s}(n,a,b)$ is at most $(1+s)\log(n+s+1)+1$.
\end{proof}

The decoding algorithm of the code $\cC_{s}(n,a,b)$ for correcting a single deletion and $s$ right-shifts of $0$ is summarized in Algorithm~\ref{alg:delshift}.

\begin{algorithm}
\label{alg:delshift}
\SetAlgoLined
\KwInput{Corrupted Sequence $\bx'$}
\KwOutput{Original Sequence $\bx\in \cC_{s}(n,a,b)$} $\Delta=\VT(\bx)-\VT(\bx')$, $c=\sum_{i=1}^n x_i-\sum_{i=1}^{|\bx'|} x'_i$ and $w=\wt(\bx')$.

\eIf{$|\bx'|=n$}
    {Recover $\bpsi(\bx)$ by $\bpsi(\bx')$ and then $\bx$.}
    {\eIf{$c=0$}{Insert a $0$ in the rightmost index of $\max(\Delta-s,0)$ $1$s to get $\hat{\bx}$. Recover $\bpsi(\bx)$ by $\bpsi(\hat{\bx})$ and then $\bx$.}
    {Insert a $1$ in the leftmost index of $\max(\Delta-w-s-1,0)$ $0$s to get $\hat{\bx}$. Recover $\bpsi(\bx)$ by $\bpsi(\hat{\bx})$ and then $\bx$.}}

\caption{Decoding procedure of $\cC_{s}(n,a,b)$}
\end{algorithm}

Further, Construction~\ref{con:del_rightshift} and Theorem~\ref{thm:del_rightshift} can be naturally extended to construct codes for correcting a single deletion, $s^{+}$ right-shifts of $0$ and $s^{-}$ left-shifts of $0$ with $s=s^{+}+s^{-}$.
\begin{corollary}\label{cor:general_shift}
    For all $a\in[0:n+s]$ and $b\in[0:1]$, the code
    \begin{equation*}
     \cC_{s^{+},s^{-}}(n,a,b)=\{ \bx\in\Sigma_2^n: \VT(\bx)\equiv a \bmod (n+s+1),\; \sum_{i=1}^n x_i \equiv b \bmod 2,\; \bpsi(\bx)\in \cC_{H}(n,2s+1)\}. 
     \end{equation*}
    can correct a single deletion, $s^{+}$ right-shifts of $0$ and $s^{-}$ left-shifts of $0$ with redundancy at most $(1+s)\log (n+s+1)+1$, where $s=s^{+}+s^{-}$.
\end{corollary}
\begin{proof}
    Similar to Proposition~\ref{pro:shift_vtchange}, assume that there are at most $s^{-}$ left-shifts of $0$s, the change of VT syndrome is $\VT(\bx)-\VT(\bx')=-s^{-}$. Assume that a $0$ is deleted using the same strategy as the proof of Theorem~\ref{thm:del_rightshift}, we also have $\Delta=R_1+k^{+}-k^{-}$, where $k^{+}$ and $k^{-}$ are actual number of right-shifts and left-shifts of $0$ occur. Also, we insert a $0$ in the index of rightmost of  $\max(\Delta-s^{+}+s^{-},0)$ $1$s to obtain $\hat{\bx}$, and denote the inserting index as $p'$. We also have $\max(R_1+k^{+}-k^{-}-s^{+}+s^{-},0)\leq R_1$. When $R_1+k^{+}-k^{-}-s^{+}+s^{-}>0$, there are $((s^{+}-s^{-})-(k^{+}-k^{-}))$ 1s between index $p$ and $p'$.
    While $R_1+k^{+}-k^{-}-s^{+}+s^{-}\leq 0$, there are $R_1$ 1s between index $p$ and $p'$, where $R_1\leq(s^{+}-s^{-})-(k^{+}-k^{-})$.
    Based on Case 1 of Lemma~\ref{lem:misinsert}, we have that there are at least $((s^{+}-s^{-})-(k^{+}-k^{-}))$ adjacent transpositions between $\hat{\bx}$ and $\bx$ and there are $k^{+}+k^{-}$ adjacent transpositions occur in $\bx$. Therefore, the total number of adjacent 
    transpositions such that $\bx$ can be obtained from $\hat{\bx}$ is at most
    \begin{equation*}
        (s^{+}-s^{-})-(k^{+}-k^{-})+(k^{+}+k^{-})=s^{+}-s^{-}+2k^{-}\leq s^{+}+s^{-}=s
    \end{equation*}
    Hence, we can recover $\bpsi(\bx)$ by $\bpsi(\hat{\bx})$ since there are at most $s$ substitutions,  and then we can recover $\bx$. Also, the analysis of redundancy is the same as that in the proof of Theorem~\ref{thm:del_rightshift}.
\end{proof}

The code in \cite{gabrys2017codes} for correcting a single deletion and $s$ adjacent transpositions requires at most $(1+2s)\log (n+2s+1)$ redundancy. Comparing our result to this, if we know that these $s$ adjacent transpositions contain $s^{+}$ right-shifts of $0$ and $s^{-}$ left-shifts of $0$, the redundancy of our code $\cC_{s^{+},s^{-}}(n,a,b)$ can be further reduced to at most $(1+s)\log (n+s+1)+1$ where $s=s^{+}+s^{-}$.

\section{List-Decodable Codes for correcting asymmetric deletions and adjacent transpositions}\label{sec:list}
In this section, we aim to construct {\em List-Decodable} codes with low redundancy. For correcting only at most $t$ $0$-deletions, Dolecek and Anatharam~\cite{dolecek2010repetition} proposed a well-known construction with optimal redundancy $t\log n+O(1)$. Inspired by this, we have the following construction:
\begin{construction}\label{con:listcode}
Let $p$ be a prime such that $p>n$. Construction $\cC_{t,s^{+},s^{-}}(n,K,\ba;p)$ is defined as the set of all $\bx \in \Sigma_2^n$ such that 
\begin{equation*}
 \sum_{i=1}^{w+1} i^{m}\phi(\bx)_{i} \equiv a_{m} \bmod p,\; \forall m\in \{1,\dotsc,K\}.
\end{equation*}
where $w=\wt(\bx)$ and $\ba=(a_1,a_2,\dotsc,a_{K})$.
\end{construction}

Let $\bx =x_1\dotsc x_n \in \Sigma_2^n$ be the original vector and $\bx'$ be the received vector after at most $t$ $0$-deletions, at most $s^{+}$ right-shifts of $0$ and at most $s^{-}$ left-shifts of $0$. Hence, we obtain the vector $\bphi(\bx')$ and the corresponding $\ba'$ at the receiver. Let $a'_m=\sum_{i=1}^{w+1} i^m \phi(\bx')_i$ and $a''_m=a_m-a'_m, \forall m\in\{1,\dotsc,K\}$.
\begin{proposition}\label{cla:trans_delta}
Assuming that there is only a single adjacent transposition occurs at the bit $0$ with index $j$ in $\bx$, the change of syndrome $a''_m$ is shown as follows:
\begin{enumerate}
    
    \item Right-shift of $0$ ($01\rightarrow 10$): \begin{equation*}
        a''_m = (j+1)^{m}-j^{m} \bmod p\\=\sum_{i=0}^{m-1}{m\choose i}j^i \bmod p\end{equation*}
    \item Left-shift of $0$ ($10\rightarrow 01$): 
    \begin{equation*}
        a''_m = j^{m}-(j+1)^{m} \bmod p= -\sum_{i=0}^{m-1}{m\choose i}j^i \bmod p\end{equation*}
\end{enumerate}
\end{proposition}

Then, we assume that $t$ $0$-deletions occur in the $0$-run before the bit $1$ with indices $(d_1,d_2,\dotsc,d_t)$, respectively, where $d_1\leq d_2\leq \dotsm \leq d_t$. Also, $s^{+}$ right-shifts of $0$ ($01\rightarrow10$) occur in the bit $0$ with indices $(j_1,j_2,\dotsc,j_{s^{+}})$ and $s^{-}$ left-shifts of $0$ ($10\rightarrow01$) occur in the bit $0$ with indices  $(k_1,k_2,\dotsc,k_{s^{-}})$, where $j_1<j_2<\dotsm<j_{s^{+}}$,  $k_1<k_2<\dotsm<k_{s^{-}}$ and   $j_v\neq k_w, \forall{v\in\{1,\dotsc,s^{+}\},w\in\{1,\dotsc,s^{-}\}}$. For simplicity, let $s=s^{+}+s^{-}$ be the number of total adjacent transpositions. 

Based on Proposition~\ref{cla:trans_delta}, considering all $t$ $0$-deletions and $s$ adjacent transpositions and set $K=t+s$, we have a set of equations showing the change of syndromes for all $m\in\{1,\dotsc,t+s\}$ as follows:

\begin{equation}\label{eq:differ_deltrans}
    a''_m\equiv\sum_{u=1}^t d_u^{m}+\sum_{i=0}^{m-1}\left[{m\choose i}\Big(\sum_{v=1}^{s^{+}} j_v^{i}-\sum_{w=1}^{s^{-}} k_w^{i}\Big)\right] \bmod p.
\end{equation}


If there are only $t$ $0$-deletions without $s$ adjacent transpositions, Dolecek and Anantharam\cite{dolecek2010repetition} showed that the following system of equations has a unique solution. 
\begin{lemma}[Dolecek and Anatharam~\cite{dolecek2010repetition}]\label{lem:lara_unique}

Let $p$ be a prime such that $p>n$. Assume the number of adjacent transpositions  $s=0$, the system of equations in \eqref{eq:differ_deltrans} can be rewritten as a set of constraints composed of $t$ distinct equations, which are presented as follows:

\begin{equation}\label{eq:tdels}
    \begin{cases}
        a_1''\equiv d_1+d_2+\dotsc+d_t \bmod p,\\
        a_2''\equiv d^2_1+d^2_2+\dotsc+d^2_t \bmod p,\\
        \vdots\\
        a_t''\equiv d^t_1+d^t_2+\dotsc+d^t_t \bmod p.\\
    \end{cases}
\end{equation} 
These equations admit a solution $d_1,d_2,\dotsc,d_t$ and $d_1\leq d_2\leq\dotsm\leq d_t$.
\end{lemma}




Following the technique in \cite{dolecek2010repetition}, if we can uniquely determine the ordered solution set $\{d_1,\ldots,d_t,j_1,\ldots,j_{s^{+}},k_1,\ldots,k_{s^{-}}\}$ of \eqref{eq:differ_deltrans}, we also can correct $t$ $0$-deletions, $s^{+}$ right-shifts of $0$, and $s^{-}$ left-shifts of $0$ with at most $(t+s)\log n$ bits of redundancy, where $s=s^{+}+s^{-}$. However, to the best of our knowledge, the result is not known to us and is still open for future work. 
    
In this section, we focus on {\em List-Decodable} code $\cC_{t,s^{+},s^{-}}(n,\kappa,\ba;p)$ for correcting at most $t$ $0$-deletions, at most $s^{+}$ right-shifts of $0$ and $s^{-}$ left-shifts of $0$. Set $K=\kappa$ in Construction~\ref{con:listcode}, where $\kappa=\max(t,s+1)$, $s=s^{+}+s^{-}$ and $p$ is a prime such that $p>n$. For the following system of equations, we can determine the solution set uniquely.

\begin{lemma}\label{lem:uni_pos&neg}
Let $p$ is a prime such that $p>n$ and $s=s^{+}+s^{-}$. A set of constraints with $s$ equations such that
\begin{equation}\label{eq:uni_pos&neg}
    \begin{cases}
        b_1''\equiv \sum_{v=1}^{s^{+}} j_v^{1}-\sum_{w=1}^{s^{-}} k_w^{1} \bmod p,\\
        b_2''\equiv \sum_{v=1}^{s^{+}} j_v^{2}-\sum_{w=1}^{s^{-}} k_w^{2} \bmod p,\\
        \vdots\\
        b_s''\equiv \sum_{v=1}^{s^{+}} j_v^{s}-\sum_{w=1}^{s^{-}} k_w^{s} \bmod p.\\
    \end{cases}
\end{equation} 
uniquely determines the set of solutions of $\{j_1,\dotsc,j_{s^{+}},k_1,\dotsc,k_{s^{-}}\}$, where $j_1<j_2<\dotsm<j_{s^{+}}$, $k_1<k_2<\dotsm<k_{s^{-}}$ and $j_v\neq k_w, \forall{v\in\{1,\dotsc,s^{+}\},w\in\{1,\dotsc,s^{-}\}}$.
\end{lemma}
We note that Lemma \ref{lem:uni_pos&neg} is similar to Lemma \ref{lem:lara_unique}. The only difference is that the coefficients of all terms in \eqref{eq:tdels} in Lemma \ref{lem:lara_unique} are positive while the coefficients of all terms in \eqref{eq:uni_pos&neg} in Lemma \ref{lem:uni_pos&neg} can be either positive or negative. We can use the same technique in Lemma \ref{lem:lara_unique} to prove Lemma \ref{lem:uni_pos&neg}.

\begin{proof}
Define the polynomials
\begin{equation*}
    \sigma^{+}(x)=\prod_{v=1}^{s^{+}}(1-j_{v}x)\quad \mbox{and}\quad \sigma^{-}(x)=\prod_{w=1}^{s^{-}}(1-k_{w}x).
\end{equation*}
Let $\sigma(x)=\sum_{m=0}^s \sigma_m x^m$ be defined by
\begin{equation*}
    \sigma(x)=\sigma^{+}(x)/\sigma^{-}(x) \bmod x^{s}
\end{equation*}
Then, we define $\sigma^{*}(x)=\sigma(x)\bmod p$ and $\sigma^{*}_m=\sigma_m\bmod p$.

We also define 
\begin{equation*}
    S^{*}(x)=\sum_{m=1}^{\infty} \Big(\sum_{v=1}^{s^{+}} j_{v}^m-\sum_{w=1}^{s^{-}} k_{w}^m\Big)x^m   \bmod p
\end{equation*}
and $S^{*}_m=\sum_{v=1}^{s^{+}} j_{v}^m-\sum_{w=1}^{s^{-}} k_{w}^m \bmod p$.

Then, we have the following generalized Newton’s identities~\cite[Lemma 10.3]{roth2006introduction} over $\mathrm{GF}(p)$ as follows
     \begin{align}
         &\sigma^{*}(x)S^{*}(x)+x(\sigma^{*}(x))'=0, \quad\text{or equivalently,}\nonumber\\
         &\sum_{m=0}^{u-1} \sigma^{*}_{m}S^{*}_{u-m}+u\sigma^{*}_{u}=0,\label{eq:newton_pos&neg}\; u\ge1.
     \end{align}
where $(\sigma^{*}(x))'$ is derivative of $\sigma^{*}(x)$.

Using a similar technique as in the proof of Lemma~\ref{lem:lara_unique}, from~\eqref{eq:newton_pos&neg}, $\sigma^{*}_m$ can be recursively obtained by $\{S^{*}_1,\dotsc,S^{*}_m\}$ and $\{\sigma^{*}_1,\dotsc,\sigma^{*}_{m-1}\}$, where $\{S^{*}_1,\dotsc,S^{*}_m\}=\{b''_1,\dotsc,b''_m\}$, which follows that all the coefficients of the polynomial $\sigma^{*}(x)=\sum_{m=0}^s \sigma_m x^m\bmod p$ are known. Further, we know that the polynomial $\sigma^{*}(x)$ has at most $s$ solutions by Lagrange's Theorem. Denote $I_0 = \{j_1, \dotsc, j_{s^{+}}, k_1, \dotsc, k_{s^{-}}\}$, where each element in $I_0$ has a value less than $p$. Additionally, we define $I_r = \{j_1 + rp, \dotsc, j_{s^{+}} + rp, k_1 + rp, \dotsc, k_{s^{-}} + rp\}$ as one of the incongruent solution sets of $I_0$. Note that $p > n$, ensuring that $I_0 \cap I_r = \emptyset$ for any $r\neq 0$. Consequently, all the incongruent solutions are distinguishable. Hence, we can conclude that the solution set $\{j_1, \dotsc, j_{s^{+}}, k_1, \dotsc, k_{s^{-}}\}$ is unique.
\end{proof}

\begin{theorem}\label{thm:list}
Let $p$ be a prime such that $p>n$. The code  $\cC_{t,s^{+},s^{-}}(n,\kappa,\ba;p)$ as defined in Construction~\ref{con:listcode} has redundancy $\kappa\log n$, where $\kappa=\max(t,s+1)$ and $s=s^{+}+s^{-}$.
If there are at most $t$ $0$-deletions, $s^{+}$ right-shifts of $0$ and $s^{+}$ left-shifts of $0$, $\cC_{t,s^{+},s^{-}}(n,\kappa,\ba;p)$ is a list-decodable code with list size $O(n^{\min(t,s+1)})$.
\end{theorem}
\begin{proof}
Let $\bx =x_1\dotsc x_n \in \Sigma_2^n$ be the original vector and $\bx'$ be the received vector after at most $t$ $0$-deletions, $s^{+}$ right-shifts of $0$ and $s^{+}$ left-shifts of $0$. Hence, we compute $\bphi(\bx')$ and $\ba'$ from $\bx'$. Also, we obtain $\ba''=\ba-\ba'$, where $\ba''=(a''_1,\dotsc,a''_{\kappa})$.

Assume $t\ge s+1$ with $s=s^{+}+s^{-}$ and expand \eqref{eq:differ_deltrans}. We have the following set of equations with $\kappa=t$:
\begin{equation}\label{eq:list_del&trans}
    \begin{cases}
        a_1''\equiv \sum_{u=1}^t d_u+(s^{+}-s^{-}) \bmod p,\\
        a_2''\equiv \sum_{u=1}^t d_u^2+(s^{+}-s^{-})+2(\sum_{v=1}^{s^{+}} j_v^{1}-\sum_{w=1}^{s^{-}} k_w^{1}) \bmod p,\\
        \vdots\\
        a_{t}''\equiv \sum_{u=1}^t d_u^{t}+(s^{+}-s^{-})+t(\sum_{v=1}^{s^{+}} j_v^{1}-\sum_{w=1}^{s^{-}} k_w^{1})\\\qquad+\dotsm+t(\sum_{v=1}^{s^{+}} j_v^{t-1}-\sum_{w=1}^{s^{-}} k_w^{t-1}) \bmod p.\\
    \end{cases}
\end{equation}

Recall that we decode uniquely if we can determine the unique set of solutions of \eqref{eq:list_del&trans}. However, we are not aware of any method to obtain the unique solution of  \eqref{eq:list_del&trans}. We know that, given $\be=(e_1,\dotsc,e_{s+1})$, we can solve the following equations uniquely.
\begin{equation}\label{eq:list}
    \begin{cases}
        e_1\equiv s^{+}-s^{-} \bmod p,\\
        e_2\equiv (s^{+}-s^{-})+2(\sum_{v=1}^{s^{+}} j_v^{1}-\sum_{w=1}^{s^{-}} k_w^{1}) \bmod p,\\
        \vdots\\
        e_{s+1}\equiv (s^{+}-s^{-})+(s+1)(\sum_{v=1}^{s^{+}} j_v^{1}-\sum_{w=1}^{s^{-}} k_w^{1})\\\qquad+\dotsm+(s+1)(\sum_{v=1}^{s^{+}} j_v^{s}-\sum_{w=1}^{s^{-}} k_w^{s}) \bmod p.\\
    \end{cases}
\end{equation}

Indeed, by denoting $\be'=(e'_1,\dotsc,e'_{s+1})$ with $me'_m=e_m-\sum_{i=1}^{m-1}\big[{m\choose {i-1}}e'_i\big]$ for all $m\in\{2,\dotsc,s+1\}$ and $e'_1=e_1$, we can rearrange \eqref{eq:list} to be similar to Lemma \ref{lem:uni_pos&neg} as follows.
\begin{equation}\label{eq:list_re}
    \begin{cases}
        e'_1\equiv s^{+}-s^{-} \bmod p,\\
        e'_2\equiv \sum_{v=1}^{s^{+}} j_v^{1}-\sum_{w=1}^{s^{-}} k_w^{1}\bmod p,\\
        \vdots\\
        e'_{s+1}\equiv \sum_{v=1}^{s^{+}} j_v^{s}-\sum_{w=1}^{s^{-}} k_w^{s} \bmod p.\\
    \end{cases}
\end{equation}
Therefore, based on Lemma~\ref{lem:uni_pos&neg}, we obtain the unique set of solutions of $\{j_1,\dotsc,j_{s^{+}},k_1,\dotsc,k_{s^{-}}\}$ from \eqref{eq:list_re}.

When the set of solution $\{j_1,\dotsc,j_{s^{+}},k_1,\dotsc,k_{s^{-}}\}$ is obtained, we compute $(e_{s+2},\dotsc,e_{t})$.
\begin{equation}\label{eq:givenvalue_2}
    e_{m}=\sum_{i=0}^{m-1}\left[{m\choose i}\Big(\sum_{v=1}^{s^{+}} j_v^{i}-\sum_{w=1}^{s^{-}} k_w^{i}\Big)\right] \bmod p.
\end{equation}
where $m\in\{s+2,\dotsc,t\}$.

Denote $\ba^{*}=(a_1^{*},\dotsc,a_t^{*})$ with $a^{*}_m=a''_m-e_m, \forall m\in\{1,\dotsc,t\}$. 
Substituting \eqref{eq:list} and \eqref{eq:givenvalue_2} into \eqref{eq:list_del&trans}, we obtain the following set of equations.
\begin{equation}\label{eq:list_del}
    \begin{cases}
        a_1^{*}\equiv \sum_{u=1}^t d_u\bmod p,\\
        a_2^{*}\equiv \sum_{u=1}^t d_u^2 \bmod p,\\
        \vdots\\
        a_{t}^{*}\equiv \sum_{u=1}^t d_u^{t} \bmod p.\\
    \end{cases}
\end{equation}
The set of equations~\eqref{eq:list_del} provides a unique set of solutions $\{d_1,\dotsc,d_t\}$ by Lemma~\ref{lem:lara_unique}. Therefore, the unique solution of all positions of $0$-deletions and adjacent transpositions $\{d_1,\dotsc,d_t,j_1,\dotsc,j_{s^{+}},k_1,\dotsc,k_{s^{-}}\}$ is obtained. 
So, for each $\be=(e_1,\ldots,e_{s+1})\in\Sigma_p^{s+1}$, we obtain the set $\{d_1,\dotsc,d_t,j_1,\dotsc,j_{s^{+}},k_1,\dotsc,k_{s^{-}}\}$. Thus, we do list-decoding with the list size $O(n^{s+1})$ since $p=O(n)$. Moreover, the size of the list-decodable code $\cC_{t,s^{+},s^{-}}(n,\kappa,\ba;p)$ with $\kappa=t$ is at least $2^n/(2n)^t$, that is, we need at most $\kappa\log n$ bits of redundancy to construct the code
$\cC_{t,s^{+},s^{-}}(n,\kappa,\ba;p)$.

When $t<s+1$, we proceed similarly to the case when $t\ge s+1$. In this case, for each given $\be\in\Sigma_p^t$, we obtain a solution set. Thus, we can do list-decoding with the list-size $O(n^t)$ with $p=O(n)$. The size of the code $\cC_{t,s^{+},s^{-}}(n,\kappa,\ba;p)$ is at least $2^n/(2n)^{s+1}$.

Then, we conclude that the list-decodable code $\cC_{t,s^{+},s^{-}}(n,\kappa,\ba;p)$ can correct at most $t$ $0$-deletions, $s^{+}$ right-shifts of $0$ and $s^{-}$ left-shifts of $0$ with list size at most $O(n^{\min(t,s+1)})$ and with redundancy $\kappa\log n+O(1)$ bits, where $s=s^{+}+s^{-}$, both $t,s$ are constant and $\kappa=\max(t,s+1)$.
\end{proof}

The decoding algorithm of the list-decodable code $\cC_{t,s^{+},s^{-}}(n,\kappa,\ba;p)$ is summarized in Algorithm~\ref{alg:listdecode}, where $s=s^{+}+s^{-}$ and $t>s+1$.

\begin{algorithm}
\label{alg:listdecode}
\SetAlgoLined
\KwInput{Corrupted Sequence $\bx'$}
\KwOutput{$O(n^{s+1})$ possible sequences, including the original codeword $\bx\in \cC_{t,s^{+},s^{-}}(n,\kappa,\ba;p)$} 
\textbf{Initialization:} Compute $\bphi(\bx')$ based on $\bx'$ and compute $\ba''$ to obtain~\eqref{eq:list_del&trans}.


\For{$\be=(e_1,\dotsc,e_{s+1})$ such that $e_i\in\{0,1,\dotsc,p-1\}$, $
\forall i\in\{1,\dotsc,s+1\}$}{
Get the solution set $\{j_1,\dotsc,j_{s^{+}},k_1,\dotsc,k_{s^{-}}\}$ by \eqref{eq:list} and \eqref{eq:list_re}.

Compute $e_m$ from the solution set $\{j_1,\dotsc,j_{s^{+}},k_1,\dotsc,k_{s^{-}}\}$ using \eqref{eq:givenvalue_2} for each $ s+2\leq  m \leq t$. 
Compute $a^*_m=a_m''-e_m$.
Solve \eqref{eq:list_del} to obtain the unique solution set $\{d_1,\dotsc,d_t\}$.
}
For each given $\be$, we can recover $\bphi(\bx)$ from $\bphi(\bx')$ by a set of error positions $\{d_1,\dotsc,d_t, j_1,\dotsc,j_{s^{+}},k_1,\dotsc,k_{s^{-}}\}$ and then output $\bx$. 
 
 \caption{List decoding procedure}
\end{algorithm}

Next, we will present a result for a special case $t=1$. 

\begin{corollary}\label{cor:list_1}
The list-decodable code $\cC_{1,s^{+},s^{-}}(n,s+1,\ba;p)$ can correct a single $0$-deletion, at most $s^{+}$ right-shifts of $0$ and $s^{-}$ left-shifts of $0$ with list size at most $s$ and redundancy $(s+1)\log n+O(1)$.
\end{corollary}
\begin{proof}
When $t=1$, the first equation of \eqref{eq:list_del&trans} can be written as $a''_1=d+(s^{+}-s^{-})\bmod p$. Then, we have $d=a''_1-(s^{+}-s^{-})\bmod p$, and the deletion position $d\in [a''_1-s^{+}:a''_1+s^{-}]$. Substituting each given value of $d$ into \eqref{eq:list_del&trans}, we have the solution set of $\{j_1,\dotsc,j_{s^{+}},k_1,\dotsc,k_{s^{-}}\}$. Therefore, the list size in this case is $s=s^{+}+s^{-}$ since $d\in [a''_1-s^{+}:a''_1+s^{-}]$.
\end{proof}

The code $\cC_{1,s^{+},s^{-}}(n,s+1,\ba;p)$ defined in Corollary~\ref{cor:list_1} is capable of correcting a single 0-deletion, at most $s^{+}$ right-shifts of $0$ and $s^{-}$ left-shifts of $0$ with constant list size at most $s$ and redundancy $(s+1)\log n+O(1)$. The list size is a constant $s$, which is significantly less than the list size $O(n)$ when we directly substitute $t=1$ to Theorem~\ref{thm:list}.

\section{Codes for correcting limited-magnitude blocks of $0$-deletions and adjacent transpositions}\label{sec:tblocks}

In this section, we focus on studying errors of \emph{$t_{\mathrm{b}}$ blocks of $0$-deletions with $\ell$-limited-magnitude and $s$ adjacent transpositions.} $t_{\mathrm{b}}$ blocks of $0$-deletions with $\ell$-limited-magnitude refer to the fact that there are at most $t_{\mathrm{b}}$ blocks of $0$s deleted with the length of each deletion block being at most $\ell$. Consequently, if $t_{\mathrm{b}}$ blocks of $0$-deletions with $\ell$-limited magnitude occur, a maximum of $t_{\mathrm{b}}\ell$ $0$s are deleted, and these $t_{\mathrm{b}}$ blocks of $0$-deletions could be distributed among at most $t_{\mathrm{b}}$ separate 0-runs. The most extreme scenario is when all $t_{\mathrm{b}}$ blocks of $0$-deletions occur within a single $0$-run, where a possible long consecutive $0$-deletions is composed of multiple blocks $0$-deletions with a limited-magnitude.

\begin{claim}\label{clm:range_change}
Given a sequence $\bx\in\Sigma_2^n$, the length of each $0$-run increases by at most $2$ and decreases by at most $t_{\mathrm{b}}\ell+2$ through $t_{\mathrm{b}}$ blocks of $0$-deletions with $\ell$-limited-magnitude and $s$ adjacent transpositions.
\end{claim}  
\begin{proof}
     Based on Claim~\ref{cla:change_trans}, it follows that the length of a $0$-run increases by $1$ when an adjacent transposition occurs. If $\bx$ contains a subsequence $01\mathbf{0}^{k}10$ and $2$ adjacent transpositions occur within this subsequence, the transformation would be $01\mathbf{0}^{k}10\rightarrow 10\mathbf{0}^{k}01$. Therefore, the length of each $0$-run increases by at most $2$. Considering that all $t_{\mathrm{b}}$ blocks of $0$-deletions with $\ell$-limited magnitude may occur within a single $0$-run, a maximum of $t_{\mathrm{b}}\ell$ 0s are deleted within this $0$-run. Additionally, if there are $2$ adjacent transpositions occurring at the start and end of this $0$-run, i.e., $10\mathbf{0}^{k}01\rightarrow01\mathbf{0}^{k}10$, the length of this $0$-run will decrease by $2$. Thus, the length of each $0$-run decreases by at most $t_{\mathrm{b}}\ell+2$. 
\end{proof}

For convenience, we append a bit $1$ at the end of $\bx$ and denote it as $\bx1$. Since the sequence $\bx1$ ends with $1$, we redefine $\bphi$ as $\bphi(\bx1)\defeq(u_1,u_2,u_3,\dotsc,u_{w})\in\mathbb{N}^{w}$, where $\bx1=0^{u_1}10^{u_2}10^{u_3}\dotsc 0^{u_{w}}1$ throughout this section.


\begin{definition}
Define the error ball
\begin{equation*}
    \cB(n,t,k_{+},k_{-})=\{{\bu\in\mathbb{Z}^n: -k_{-}\leq u_i\leq k_{+}}, \wt(\bu)\leq t, \sum_{i=1}^n u_i\leq 0\}.
\end{equation*}
\end{definition}

For a given sequence $\bx\in\Sigma_2^n$, if \emph{$t_{\mathrm{b}}$ blocks of $0$-deletions with $\ell$-limited-magnitude and $s$ adjacent transpositions} occur, then we obtain the sequence  $\bx'$ such that $\bphi(\bx'1)=\bphi(\bx1)+\bv$, where $\bv\in \cB(w,t_{\mathrm{b}}+2s,2,t_{\mathrm{b}}\ell+2)$ and $w=\wt(\bx'1)=\wt(\bx1)$. Since each adjacent transposition could affect the length of two $0$-runs, and considering that $t_{\mathrm{b}}$ blocks of $0$-deletions could occur in at most $t_{\mathrm{b}}$ $0$-runs, we can deduce that $\wt(\bv) \leq t_{\mathrm{b}} + 2s$. Besides, the range of the value of $v_i$ has been shown in Claim~\ref{clm:range_change}.

\begin{example}
    Assume that $\bx=0100101001\in\Sigma_2^{10}$ with $\ell=2$, $t_{\mathrm{b}}=3$ and $s=1$, then $\bphi(\bx1)=12120$. If the received sequence $\bx'=0110110\in\Sigma_2^{6}$ and the corresponding $\bphi(\bx'1)=10101$, by comparing $\bphi(\bx1)$ and $\bphi(\bx'1)$, we have $\bv=(0,-2,0,-2,1)\in\cB(5,5,2,8)$.
\end{example}

Denote the set $\Phi=\{\bphi(\bx1):\bx\in\Sigma_2^n\}$ and $\Sigma_2^n$ is the set containing all binary sequences with length $n$. 
\begin{lemma}\label{lem:size_phix}
     The cardinality of $\Phi$ is:
    \begin{equation}
    |\Phi|=\sum_{w=1}^{n+1}{n \choose w-1}=2^n.
    \end{equation}
\end{lemma}
\begin{proof}
    For a binary sequence $\bx\in\Sigma_2^n$, the corresponding sequence $\bphi(\bx1)$ has length $w=\wt(\bx1)$ and $\wt(\bphi(\bx1))=n+1-w$. The
    cardinality of $\Phi$ is equivalent to the number of ways of arranging $n+1-w$ indistinguishable objects in $w$ distinguishable boxes. Also, the function $\bphi$ is a one-to-one mapping function, the cardinality of $\Phi$ should be the same as $|\Sigma_2^n|=2^n$. Thus, we get the cardinality of $\Phi$ as shown in Lemma~\ref{lem:size_phix}.  
\end{proof}

\subsection{Non-systematic Code Construction}
In this section, we will provide a non-systematic construction for a code capable of correcting $t_{\mathrm{b}}$ blocks of $0$-deletions with $\ell$-limited-magnitude and $s$ adjacent transpositions. Also, we present a decoding algorithm of this code and a lower bound on the code size.

\begin{lemma}
	Let $p$ be the smallest prime larger than $t_{\mathrm{b}}\ell+4$. The family of codes 
	\begin{equation*}
    \cC(n,t_{\mathrm{b}},\ell,s)=\bigcup\limits_{w=1}^{n+1}\{\bx\in\Sigma_2^n:\bphi(\bx1)\bmod p\in\cC_{p}, \wt(\bx1)=w\},
	\end{equation*}
    is capable of correcting $t_{\mathrm{b}}$ blocks of $0$-deletions with $\ell$-limited-magnitude and $s$ adjacent transpositions for $\bx\in \cC(n,t_{\mathrm{b}},\ell,s)$ if the code $\cC_{p}$ defined over $\Sigma_{p}^{w}$ is capable of correcting $t_{\mathrm{b}}+2s$ substitutions for $\bphi(\bx1)$.
\end{lemma}

\begin{proof}

For $\bx, \by \in \cC(n, t_{\mathrm{b}}, \ell, s)$ and $\bv, \bv' \in \cB(w, t_{\mathrm{b}} + 2s, 2, t_{\mathrm{b}}\ell + 2)$ with $w = \wt(\bx1)$, we assume that $\bphi(\bx1) + \bv = \bphi(\by1) + \bv'$. Then, $\bv' - \bv = \bphi(\by1) - \bphi(\bx1)$, and we have $\bv'' = (\bv' - \bv) \bmod p \in \cC_p$ since both $\bphi(\bx1), \bphi(\by1) \bmod p \in \cC_p$. Based on the definition of $\cB(w, t_{\mathrm{b}} + 2s, 2, t_{\mathrm{b}}\ell + 2)$, we have $\wt(\bv), \wt(\bv') \leq t_{\mathrm{b}} + 2s$ and $\wt(\bv'') \leq 2(t_{\mathrm{b}} + 2s)$. Since $\cC_p$ is a code capable of correcting $t_{\mathrm{b}} + 2s$ substitutions, we conclude that $\bv'' = \mathbf{0}$. Furthermore, each entry in $(\bv - \bv')$ takes a value within the range of $[-(t_{\mathrm{b}}\ell + 4) : (t_{\mathrm{b}}\ell + 4)]$. It is sufficient to show that $\bv = \bv'$, which holds since $p$ is the smallest prime larger than $t_{\mathrm{b}}\ell + 4$. Finally, this implies $\bphi(\bx1) = \bphi(\by1)$ and $\bx = \by$. 
\end{proof}

\begin{lemma}(\cite{aly2007quantum}, Theorem~10 )\label{lem:prime_bch}
    Let $p$ be a prime such that the distance $2\leq d\leq p^{\lceil m/2\rceil-1}$ and $n=p^m-1$. Then, there exists a narrow-sense $[n,k,d]$-BCH code $\cC_{p}$ over $\Sigma_{p}^n$ with
    \begin{equation*}
        n-k=\lceil (d-1)(1-1/p)\rceil m.
    \end{equation*}
\end{lemma}

\begin{theorem}\label{thm:bchblocks}
    Let $p$ be the smallest prime such that $p\ge t_{\mathrm{b}}\ell+4$, $w=p^m-1$, $w=\wt(\bx1)$ and $\cC_{p}$ is a primitive narrow-sense $[w,k,2(t_{\mathrm{b}}+2s)+1]$-BCH code with $w-k=\lceil 2(t_{\mathrm{b}}+2s)(1-1/p)\rceil m$. The family of codes 
    \begin{equation}\label{eq:nonsystem}
        \cC(n,t_{\mathrm{b}},\ell,s)=\bigcup\limits_{w=1}^{n+1}\{\bx\in\Sigma_2^n:\bphi(\bx1)\bmod p\in\cC_{p}, \wt(\bx1)=w\}.
    \end{equation}
    is capable of correcting $t_{\mathrm{b}}$ blocks of $0$-deletions with $\ell$-limited-magnitude and $s$ adjacent transpositions.
\end{theorem}
\begin{proof}
Let $\bx\in\cC(n,t_{\mathrm{b}},\ell,s)$ be a codeword, and $\bx'$ be the output through the channel that induces $t_{\mathrm{b}}$ blocks of $0$-deletions with $\ell$-limited-magnitude and $s$ adjacent transpositions. Let $\bz'=\bphi(\bx'1) \bmod p$, where $p$ is the smallest prime larger than $t_{\mathrm{b}}\ell+4$. Run the decoding algorithm of $\cC_p$~\cite{roth1988encoding} on $\bz'$ and output $\bz^{*}$. Thus, $\bz^{*}$ is also a codeword of the linear code $\cC_{p}$ and it can be shown that $\bz^{*}=\bphi(\bx1) \bmod p$. Denote $\bepsilon'=(\bz'-\bz^{*}) \bmod p$, we have that
\begin{equation}
    (\bphi(\bx'1)-\bphi(\bx1)) \bmod p=(\bz'-\bz^{*}) \bmod p=\bepsilon'.
\end{equation}
and the error vector $\bepsilon$ satisfies
\begin{equation}
   \epsilon_i=\begin{cases}
        \epsilon'_i, &{\text{if}}\; 0\leq \epsilon'_i\leq 1 \\
        \epsilon'_i-p, &{\text{otherwise}}
    \end{cases}.
\end{equation}
Hence, we can calculate $\bphi(\bx1)=\bphi(\bx'1)-\bepsilon$ and then recover $\bx$ from $\bphi(\bx1)$.
\end{proof}

The detailed decoding steps are shown in Algorithm~\ref{alg:decalg1}.

\begin{algorithm}[h]\label{alg:decalg1}
\SetAlgoLined
\KwInput{Retrieved sequence $\bx'$}
\KwOutput{Decoded sequence $\bx\in\cC(n,t_{\mathrm{b}},\ell,s)$.}

\textbf{Initialization:}  Let $p$ be the smallest prime larger than $t_{\mathrm{b}}\ell+4$. Also, append $1$ at the end of $\bx'$ and get $\bphi(\bx'1)$. 

\textbf{Step 1:} $\bz'=\bphi(\bx'1) \bmod p$. Run the decoding algorithm of $\cC_{p}$ on $\bz'$ to get the output $\bz^{*}$.

\textbf{Step 2:} $\bepsilon'=(\bz'-\bz^{*}) \bmod p$ and then $\bepsilon$.  $\bphi(\bx1)=\bphi(\bx'1)-\bepsilon$.

\textbf{Step 3:} Calculate $\bx1=\bphi^{-1}(\bphi(\bx1))$ and then recover $\bx$.

 \caption{Decoding Algorithm of $\cC(n,t_{\mathrm{b}},\ell,s)$}
 \end{algorithm}

\begin{example}
Assume $\bx=0100101001$ and $\bx'=0110110\in\Sigma_2^{6}$ with $\ell=2$, $t_{\mathrm{b}}=2$ and $s=1$. Since the retrieved sequence $\bx'=0110110$, then $\bphi(\bx'1)=10101$ and $\bz'=\bphi(\bx')\bmod 11=10101$, where $p=11$ is smallest prime such that $p\ge t_{\mathrm{b}}\ell+4$.

Run the decoding algorithm of $\cC_{p}$ on $\bz'\in\cC_{p}$, we have the output sequence $\bz^{*}=12120$. Hence $\bepsilon'=(\bz'-\bz^{*})\bmod 11=(0,9,0,9,1)$ and $\bepsilon=(0,-2,0,-2,1)$. Thus, the output of the decoding algorithm $\bphi(\bx1)=\bphi(\bx'1)-\bepsilon=(1,0,1,0,1)-(0,-2,0,-2,1)=(1,2,1,2,0)$. Finally, $\bx1=01001010011$ and $\bx=0100101001$.
\end{example}




Next, we will present a lower bound on the size of $\cC(n,t_{\mathrm{b}},\ell,s)$.
\begin{theorem}
	Let $p$ be the smallest prime larger than $t_{\mathrm{b}}\ell+4$. The size of the code $\cC(n,t_{\mathrm{b}},\ell,s)$ defined in \eqref{eq:nonsystem} in Theorem~\ref{thm:bchblocks} is lower bounded as
    \begin{equation*}
        |\cC(n,t_{\mathrm{b}},\ell,s)|\ge \frac{2^n}{p(n+1)^{\lceil2(t_{\mathrm{b}}+2s)(1-1/p)\rceil}}.
    \end{equation*}
\end{theorem}
\begin{proof}
    Denote $\bz=\bphi(\bx1)\bmod p$. $\bphi(\bx1)$ can be written as $\bphi(\bx1)=\bz+p\cdot \ba$, where $\ba$ is a vector with the same length as $\bphi(\bx1)$ and $\bz$. 
    Furthermore, since $\bz \in \cC_{p}$ and $\cC_{p}$ is a linear code, the code $\cC_{p}$ with length $w$ is a set obtained by partitioning $\Sigma_{p}^{w}$ into $p^{w-k}$ classes.
    
    Denote $\Phi_{w}$ as the set of all vectors $\bphi(\bx1)$ with length $w$. For any fixed weight $w$, we define the set $\Phi_{\bz}=\{\bz\in \Phi_{w}: \bz \bmod p\in\cC_p\}$
    \begin{equation*}
    \left|\Phi_{\bz}\right|=\frac{{n \choose w-1}}{p^{w-k}}.
    \end{equation*}
    Then, the size of the code $\cC(n,t_{\mathrm{b}},\ell,s)$ in Theorem~\ref{thm:bchblocks} can be lower bounded as follows:
    
    \begin{align}\label{eq:size_c}
|\cC(n,t_{\mathrm{b}},\ell,s)|=\sum_{w=1}^{n+1}\left|\Phi_{\bz}\right|&=\sum_{w=1}^{n+1}\left[\frac{{n \choose w-1}}{p^{w-k}}\right]\nonumber\\
&\ge \frac{\sum_{w=1}^{n+1}{n \choose w-1}}{p^{n+1-k}}=\frac{2^n}{p^{n+1-k}}.
\end{align} 
From Lemma~\ref{lem:prime_bch} and Theorem~\ref{thm:bchblocks}, let $d=2(t_{\mathrm{b}}+2s)+1$ and $m=\log_p(n+1)$.
\begin{equation}
\label{eq:p_nk}
    p^{n-k+1}=p^{\lceil2(t_{\mathrm{b}}+2s)(1-1/p)\rceil\cdot \log_{p}(n+1)+1}=p(n+1)^{\lceil2(t_{\mathrm{b}}+2s)(1-1/p)\rceil}.
\end{equation}

Therefore, from \eqref{eq:size_c} and \eqref{eq:p_nk}, the size of the code $\cC(n,t_{\mathrm{b}},\ell,s)$ in Theorem~\ref{thm:bchblocks} is lower bounded as follows:
    \begin{equation*}
        |\cC(n,t_{\mathrm{b}},\ell,s)|\ge \frac{2^n}{p(n+1)^{\lceil2(t_{\mathrm{b}}+2s)(1-1/p)\rceil}},
    \end{equation*}
    where $p$ is the smallest prime larger than $t_{\mathrm{b}}\ell+4$. \qedhere
\end{proof}

\subsection{Systematic Code Construction}
In the previous subsection, we propose a non-systematic code $\cC(n,t_{\mathrm{b}},
\ell,s)$ for correcting $t_{\mathrm{b}}$ blocks of $0$-deletions with $\ell$-limited-magnitude and $s$ adjacent transpositions. In this subsection, we will specify the efficient encoding and decoding function based on the code $\cC(n,t_{\mathrm{b}},
\ell,s)$ presented in Theorem~\ref{thm:bchblocks}.

\subsubsection{Efficient Encoding}

Before providing the efficient systematic encoding algorithm, we now introduce a mapping function proposed in \cite{knuth1986efficient} by Knuth for encoding balanced sequences efficiently. A balanced sequence is a binary sequence with an equal number of $0$s and $1$s, which will be used for identifying the information bits and redundant bits in our proposed code.

\begin{definition}
   For any $\bx\in\Sigma_2^k$, define the function $\bs:\Sigma_2^k\rightarrow \Sigma_2^{n}$ such that $\bs(\bx)\in\Sigma_2^{n}$ whose first bit is $1$ and $\bs(\bx)_{[2:n]}$ is a balanced sequence with $\lceil(n-1)/2\rceil$ $0$s and $\lfloor(n-1)/2\rfloor$ $1$s, where $n=k+\log k+1$.
\end{definition}

Next, we define a function that converts a non-binary string into a binary one.
\begin{definition}
For any $\bu\in\Sigma_p^k$, define the function  $\bb:\Sigma_p^k\rightarrow \Sigma_2^{n}$ such that $\bb(\bu)_{[i\cdot\lceil\log p\rceil+1:(i+1)\cdot\lceil\log p\rceil]}$ is the binary form of $u_i$, where $n=k\cdot\lceil\log p\rceil $.    
\end{definition}

An adjacent transposition can be regarded as two substitutions, hence the maximum number of deletions and substitutions in the $t_{\mathrm{b}}$ blocks of $0$-deletions with $\ell$-limited-magnitude and $s$ adjacent transpositions is $r=t_{\mathrm{b}}\ell+2s$.  The following lemma is used for correcting deletions, insertions and substitutions up to $r=t_{\mathrm{b}}\ell+2s$ in a binary sequence. 
\begin{lemma}[Sima,~Gabrys,~and Bruck\cite{sima2020optimalbinary}]\label{lem:binary_r}
 Let $t_{\mathrm{b}},\ell,s$ be constants with respect to $k$. There exist an integer $a\leq 2^{2r\log k+o(\log k)}$ and a function $f_{r}:\Sigma_2^k\rightarrow\Sigma_{2^{\cR_{r}(k)}}$, where $\cR_{r}(k)=O(r^4\log k)$ such that $\{(\bx,\bb(a),\bb(f_{r}(\bx)\bmod a)):\bx\in\Sigma_2^k\}$ can correct deletions, insertions and substitutions up to $r=t_{\mathrm{b}}\ell+2s$. 
\end{lemma}
For simplicity, we define the redundant bits of the code in Lemma~\ref{lem:binary_r} as $g_r(\bx)=(\bb(a),\bb(f_{r}(\bx)\bmod a))\in\Sigma_2^{4r\log k+o(\log k)}$ for given $\bx\in\Sigma_2^k$, where $r=t_{\mathrm{b}}\ell+2s$. Besides, given the parameters $t_{\mathrm{b}}$, $\ell$ and $s$, we provide the definition of the redundant bits of BCH code in Lemma~\ref{lem:prime_bch}.
\begin{definition}
Let $p$ be the smallest prime larger than $t_{\mathrm{b}}\ell+4$. Given the input $\bx\in\Sigma_2^k$, denote $g(\bx)$ as the redundant bits such that $(\bx,g(\bx))$ is a codeword of the $p$-ary primitive narrow-sense $[n,k,2(t_{\mathrm{b}}+2s)+1]$-BCH codes in Lemma~\ref{lem:prime_bch}, where $n=k+\lceil2(t_{\mathrm{b}}+2s)(1-1/p)\rceil m$ and $n=p^m-1$.
\end{definition}

Assuming the input sequence is $\bc \in \Sigma_2^k$, let us consider $\bphi(\bc1)$ with length $r_{c} = \wt(\bc1)$. We can compute $\bc' = \bphi(\bc1) \bmod p \in \Sigma_{p}^{r_{c}}$, where $p$ is the smallest prime larger than $t_{\mathrm{b}}\ell+4$. Then, we append $\mathbf{0}^{k+1-r_c}$ at the end of $\bc'$ and define $\bar{\bc} = (\bc',\mathbf{0}^{k+1-r_c}) \in \Sigma_p^{k+1}$.

Next, we encode $\bar{\bc}$ using the $p$-ary primitive narrow-sense BCH code $\cC_p$ in Lemma~\ref{lem:prime_bch}. The code takes an input of length $k+1$ and has a minimum distance of $d = 2(t_{\mathrm{b}}+2s)+1$. The output consists of redundant bits denoted as $g(\bar{\bc})$. Notably, since the input length is $k+1$, the length of $g(\bar{\bc})$ should be $\lceil 2(t_{\mathrm{b}}+2s)(1-1/p) \rceil m + 1$ in order to satisfy $n-(k+1) < n-k \leq \lceil 2(t_{\mathrm{b}}+2s)(1-1/p) \rceil m$. Subsequently, we convert the redundant bits $g(\bar{\bc})$ into a binary sequence denoted as $\bb(g(\bar{\bc}))$ and further transform $\bb(g(\bar{\bc}))$ into a balanced sequence $\bs(\bb(g(\bar{\bc})))$. Additionally, we add two leading $1$s as protecting bits at the start of $\bs(\bb(g(\bar{\bc})))$, denoting the resulting sequence as $h_1(\bar{\bc}) = (1,1,\bs(\bb(g(\bar{\bc}))))$.

In addition, we need to protect the redundant bits $h_1(\bar{\bc})$. The approach is to apply the code from Lemma~\ref{lem:binary_r} to $h_1(\bar{\bc})$, as it can correct up to $t_{\mathrm{b}}\ell+2s$ deletions and substitutions. We denote the output as $g_{r}(h_1(\bar{\bc}))$. Then, we convert $g_{r}(h_1(\bar{\bc}))$ into a balanced sequence denoted as $\bs(g_{r}(h_1(\bar{\bc})))$ and repeat each bit of its bit $2t_{\mathrm{b}}\ell+3$ times. We define $h_2(\bar{\bc}) = \mathrm{Rep}_{2t_{\mathrm{b}}\ell+3}(\bs(g_{r}(h_1(\bar{\bc}))))$, where $\mathrm{Rep}_{k}\bx$ represents the $k$-fold repetition of $\bx$.

Finally, we obtain the output $\mathrm{Enc}(\bc)=(\bc, h(\bc))$, where $h(\bc)=(h_1(\bar{\bc}),h_2(\bar{\bc}))$.  The detailed encoding steps are summarized in the following Algorithm~\ref{alg:enc} and the encoding process is also illustrated in Fig.~\ref{fig:ill_encoding}. Given a sequence $\bc\in\Sigma_2^{k}$, Algorithm~\ref{alg:enc} outputs an encoded sequence $\mathrm{Enc}(\bc)\in\Sigma_2^N$ that is capable of correcting $t_{\mathrm{b}}$ blocks of $0$-deletions with $\ell$-limited-magnitude and $s$ adjacent transpositions.

\begin{figure}
    \centering
    \includegraphics[width=15cm]{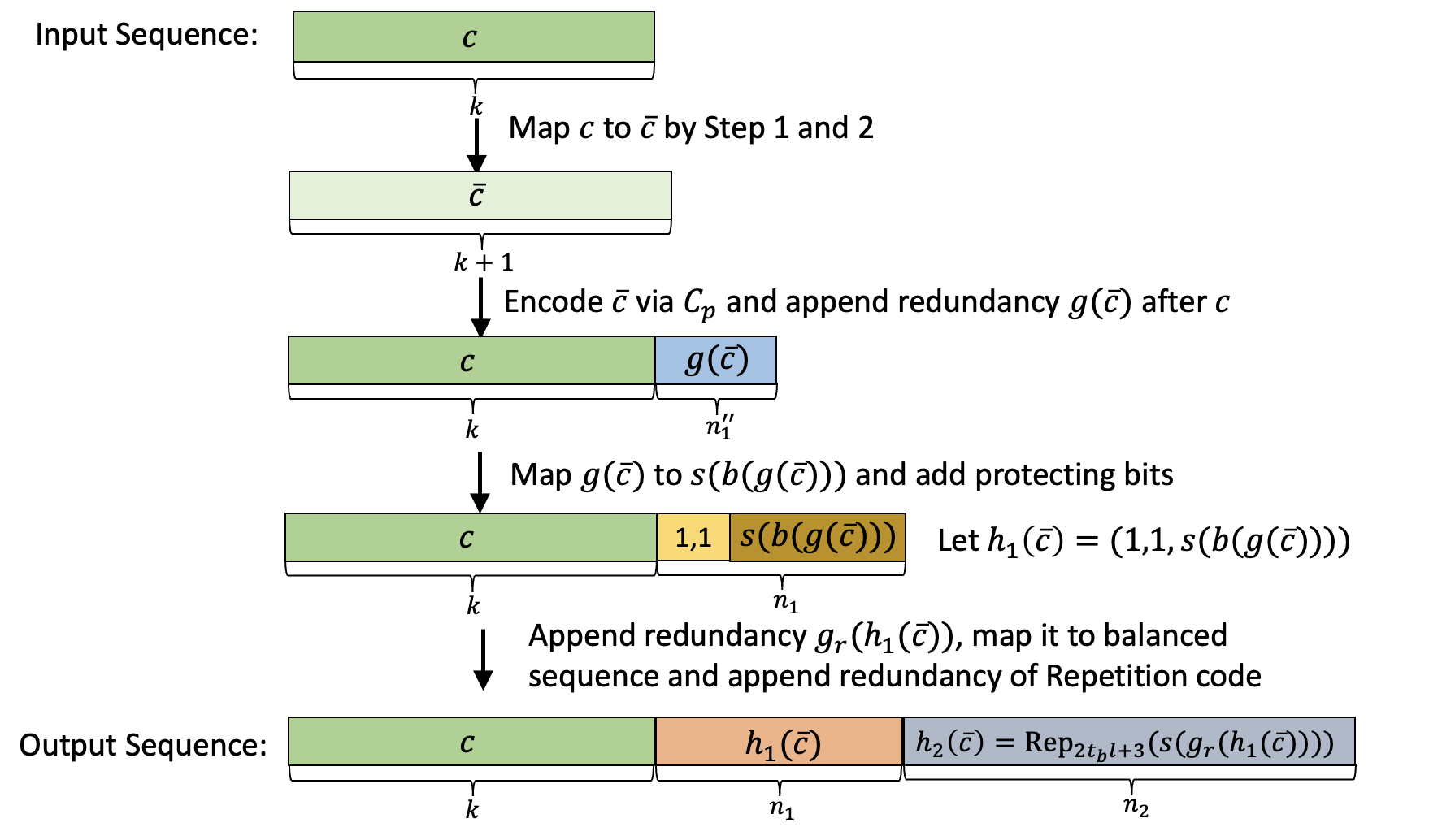}
    \caption{Illustration of the Encoding Algorithm} \label{fig:ill_encoding}
\end{figure}

\begin{algorithm}[h]\label{alg:enc}
\SetAlgoLined
\KwInput{$\bc\in \Sigma_2^k$}
\KwOutput{Encoded sequence $\mathrm{Enc}(\bc)\in\Sigma_2^N$}

\textbf{Initialization:}  Let $p$ be the smallest prime larger than $t_{\mathrm{b}}\ell+4$. 

\textbf{Step 1:} Append $1$ at the end of $\bc$ and get $\bphi(\bc1)$ with length $r_c=\wt(\bc1)$.

\textbf{Step 2:} Calculate $\bc'=\bphi(\bc1)\bmod p\in\Sigma_p^{r_c}$ and append $\mathbf{0}^{k+1-r_c}$ at the end of $\bc'$. Then, denote $\bar{\bc}=(\bc',\mathbf{0}^{k+1-r_c})$.

\textbf{Step 3:} Encode $\bar{\bc}$ via $\cC_{p}$ and output $g(\bar{\bc})$. Mapping 
$g(\bar{\bc})$ to balanced binary sequence $\bs(\bb(g(\bar{\bc})))$ and introduce protecting bits. Denote $h_1(\bar{\bc})=(1,1,\bs(\bb(g(\bar{\bc}))))$. 

\textbf{Step 4:} Protect $h_1(\bar{\bc})$ via $g_{r}$ and repeat each bit $2t_{\mathrm{b}}\ell+3$ times to obtain $h_2(\bar{\bc})=\mathrm{Rep}_{2t_{\mathrm{b}}\ell+3}\bs(g_{r}(h_1(\bar{\bc})))$. Denote the total redundant bits $h(\bc)=(h_1(\bar{\bc}),h_2(\bar{\bc}))$.
    
\textbf{Step 5:} Output $\mathrm{Enc}(\bc)=(\bc,h(\bc))\in\Sigma_2^N$.
 \caption{Encoding Algorithm}
 \end{algorithm}

Therefore, the total redundant bits $h(\bc)=(h_1(\bar{\bc}),h_2(\bar{\bc}))$  of the codeword through this encoding process is shown as follows.
\begin{theorem}\label{thm:totalred}
Let $p$ be the smallest prime such that $p\ge t_{\mathrm{b}}\ell+4$. The total redundant bits $h(\bc)$ of the codeword $\mathrm{Enc}(\bc)\in\Sigma_2^N$ by given input $\bc\in\Sigma_2^k$ is
   \begin{equation*}
       N-k= \frac{\lceil2(t_{\mathrm{b}}+2s)(1-1/p)\rceil\cdot\lceil\log p\rceil}{\log p} \log(N+1)\\+O(\log\log N).
    \end{equation*}
\end{theorem}

\begin{proof}

    Let $m=\log_p(N+1)$, hence $N=p^m-1$.  The lengths of the redundant bits are as follows:
\begin{itemize}
    \item $n''_1$ is the length of $g(\bar{\bc})$: $n''_1=\lceil2(t_{\mathrm{b}}+2s)(1-1/p)\rceil m + 1$;
    \item $n'_1$ is the length of $\bb(g(\bar{\bc}))$: $n'_1=n''_1\cdot\lceil\log p\rceil$;
    \item $n_1$ is the length of $h_1(\bar{\bc})$: $n_1=n'_1+\log n'_1+3$;
    \item $n''_2$ is the length of $g_{r}(h_1(\bar{\bc}))$: $n''_2=4(t_{\mathrm{b}}\ell+2s)\log n_1+\log n_1$;
    \item $n'_2$ is the length of $\bs(g_r(h_1(\bar{\bc})))$: $n'_2=n''_2+\log n''_2+1$;
    \item $n_2$ is the length of $h_2(\bar{\bc})$: $n_2=(2t_{\mathrm{b}}\ell+3)n'_2$;
    
\end{itemize}
    Based on the above statement, we can see that $N-k=n_1+n_2$, where
    \begin{equation*}
        n'_1=(\lceil2(t_{\mathrm{b}}+2s)(1-1/p)\rceil m)\cdot\lceil\log p\rceil
    \end{equation*}
with $m=\log_{p}(N+1)$. Hence, we have
\begin{equation*}
        n'_1=\frac{\lceil2(t_{\mathrm{b}}+2s)(1-1/p)\rceil\cdot\lceil\log p\rceil}{\log p} \log(N+1)
    \end{equation*}
Since both $t_{\mathrm{b}}$, $p$ and $s$ are constants, then $\log n'_1=O(\log\log N)$ and $n_2=O(\log\log N)$.
Therefore, the total redundant bits of the codeword $\mathrm{Enc}(\bc)\in\Sigma_2^N$ given the input $\bc\in\Sigma_2^k$ is shown as the Theorem~\ref{thm:totalred}.
\end{proof}

\subsubsection{Decoding Algorithm}

Without loss of generality, let us assume that the encoded sequence $\mathrm{Enc}(\bc) \in \Sigma_2^N$ is transmitted through \textit{exactly} $t_{\mathrm{b}}$ blocks of $0$-deletions with $\ell$-limited magnitude and \textit{exactly} $s$ adjacent transpositions in the channel. As a result, we obtain the retrieved sequence $\bd \in \Sigma_2^{N-t_{\mathrm{b}}\ell}$. In this subsection, we will introduce the decoding algorithm for obtaining $\mathrm{Dec}(\bd) \in \Sigma_2^{k}$ given $\bd \in \Sigma_2^{N-t_{\mathrm{b}}\ell}$.

The error type is at most $t_{\mathrm{b}}$ blocks of $0$-deletions with $\ell$-limited-magnitude and $s$ adjacent transpositions, the number of $1$s in $\bd$ is the same as that of in $\mathrm{Enc}(\bc)$. To identify the start of redundant bits in $\bd$, we can count the number of $1$s starting from the end of the sequence since the redundant bits is a balanced string and the number of $1$s does not change through $0$-deletions and adjacent transpositions. Hence, we find the $\lfloor(n_2+2t_{\mathrm{b}}\ell+3)/2\rfloor$-th $1$ and $(\lfloor n_1/2+n_2/2\rfloor+t_{\mathrm{b}}\ell+3)$-th $1$s from the end of $\bd$ and denote their entries as $i_{r2}$ and $i_{r1}$, respectively. For the subsequence $\bd_{[i_{r2}:N-t_{\mathrm{b}}\ell]}$, considering that there are at most $t_{\mathrm{b}}\ell$ $0$ deletions and $s$ adjacent transpositions occurring in $\mathrm{Enc}(\bc)_{[N-n_2+1:N]}$, we can utilize the $(2t_{\mathrm{b}}\ell+3)$-fold repetition code to recover $\bs(g_r(h_1(\bar{\bc})))$. Then, we can obtain the redundant bits $g_r(h_1(\bar{\bc}))$.

For the subsequence $\bd_{[i_{r1}:i_{r2}-1]}$, there are also at most $t_{\mathrm{b}}\ell$ $0$-deletions and $2s$ substitutions occurring in $\mathrm{Enc}(\bc)_{[N-n_1-n_2+1:N-n_2]}$. The recovered redundant bits $g_r(h_1(\bar{\bc}))$ can be used to recover $h_1(\bar{\bc})$. Moreover, we remove the two protecting bits of $1$s from $h_1(\bar{\bc})$ and obtain $g(\bar{\bc})$ from $h_1(\bar{\bc})=\bs(\bb(g(\bar{\bc})))$.

Finally, we define $\bz=(\bphi(\bd_{[1:i_{r1}-1]},1),\mathbf{0}^{k+1-r_c})$ and $\bz'=\bz\bmod p$, where $r_c$ represents the length of $\bphi(\bd_{[1:i_{r1}-1]},1)$ and $k=N-n_1-n_2$. The decoding steps that follow are the same as those outlined in Algorithm~\ref{alg:decalg1}, with the input for Step 1 being $\bz'$. The decoding process of $\bz'$ is assisted by the recovered redundant bits $g(\bar{\bc})$ from the BCH code $\cC_{p}$. Prior to the final step of $\bphi^{-1}$, we need to remove $\mathbf{0}^{k+1-r_c}$ at the end. The main steps for decoding the retrieved sequence $\bd\in\Sigma_2^{N-t_{\mathrm{b}}\ell}$ are summarized in Algorithm~\ref{alg:dec2}.

\begin{algorithm}[h]\label{alg:dec2}
\SetAlgoLined
\KwInput{Retrieved sequence $\bd$}
\KwOutput{Decoded sequence $\mathrm{Dec}(\bd)\in\Sigma_2^k$}

\textbf{Initialization:}  Let $p$ be the smallest prime larger than $t_{\mathrm{b}}\ell+4$. 

\textbf{Step 1:} Find the $\lfloor(n_2+2t_{\mathrm{b}}\ell+3)/2\rfloor$-th $1$ and $(\lfloor n_1/2+n_2/2\rfloor+t_{\mathrm{b}}\ell+3)$-th $1$ from the end of $\bd$ and denote their entries as $i_{r2}$ and $i_{r1}$, respectively.

\textbf{Step 2:} Recover $\bs(g_r(h_1(\bar{\bc})))$ from $\bd_{[i_{r2}:N-t_{\mathrm{b}}\ell]}$ via the Repetition code and then get $g_r(h_1(\bar{\bc}))$.

\textbf{Step 3:} Recover $h_1(\bar{\bc})$ via $g_r(h_1(\bar{\bc}))$ and then obtain $h_1(\bar{\bc})$. Furthermore, obtain $g(\bar{\bc})$ from $h_1(\bar{\bc})=\bs(\bb(g(\bar{\bc})))$.

\textbf{Step 4:} 
Let $\bz=(\bphi(\bd_{[1:i_{r1}-1]},1),\mathbf{0}^{k+1-r_c})$ and $\bz'=\bz\bmod p$. Input $\bz'$ into Step 1 of Algorithm~\ref{alg:decalg1} and run the remaining steps of Algorithm~\ref{alg:decalg1}. $\bz'$ is decoded via the recovered redundant bits $g(\bar{\bc})$ of the BCH code $\cC_{p}$.
    
\textbf{Step 5:} Output $\mathrm{Dec}(\bd)$.

\caption{Decoding Algorithm}
\end{algorithm}

\subsubsection{Time Complexity}

For the encoding algorithm, given  constants $t_{\mathrm{b}},\ell$ and $s$, each codeword is generated by following steps:
\begin{itemize}
\item First, given an input binary message string $\bc$, $\bar{\bc}$ is generated with a time complexity of $O(n)$.
\item Second, encode $\bar{\bc}$ using the $p$-ary narrow-sense BCH code $\cC_{p}$, which has a time complexity of $O((t_{\mathrm{b}}+2s)n\log n)$~\cite{roth1988encoding}.
\item Third, map the labeling function of $\bar{\bc}$ to a balanced binary sequence with a time complexity of $O(\log n)$~\cite{knuth1986efficient}.
\item 
Fourth, protect the redundant bits $h_1(\bar{\bc})$ using the code in Lemma~\ref{lem:binary_r} and convert it into a balanced sequence with time complexity of $O((\log n)^{2(t_{\mathrm{b}}\ell+2s)+1})$~\cite{sima2020optimalbinary}.
\end{itemize}
Therefore, the time complexity of the encoder time complexity is dominated by the encoding of the $p$-ary narrow-sense BCH code and the code in Lemma~\ref{lem:binary_r}, which is $O((t_{\mathrm{b}}+2s)n\log n+(\log n)^{2(t_{\mathrm{b}}\ell+2s)+1})$.

The decoding algorithm is comprised of reverse processes of the encoding steps. Therefore, it can be easily shown that the time complexity is primarily determined by decoding the $p$-ary narrow-sense BCH code and the code in Lemma~\ref{lem:binary_r}.
The decoding complexity of the $p$-ary narrow-sense BCH code is $O((t_{\mathrm{b}}+2s)n)$~\cite{roth1988encoding} and the decoding complexity of the code in Lemma~\ref{lem:binary_r} is $O((\log n)^{t_{\mathrm{b}}\ell+2s+1})$~\cite{sima2020optimalbinary}. Therefore, the total time complexity of decoding is $O((t_{b}+2s)n+(\log n)^{t_{\mathrm{b}}\ell+2s+1})$.

\section{Conclusion}\label{sec:conclusion}
In this paper, motivated by the errors in the DNA data storage and flash memories, we presented codes for correcting asymmetric deletions and adjacent transpositions. We first present three uniquely-decodable codes for different types of asymmetric deletions and adjacent transpositions. We then construct a list-decodable code for correcting asymmetric deletions and adjacent transpositions with low redundancy. Finally, we present the code for correcting $t_{\mathrm{b}}$ blocks of $0$-deletions with $\ell$-limited-magnitude and $s$ adjacent transpositions.

However, there still remain some interesting open problems.
 \begin{itemize}
     \item Construct codes that are capable of correcting \emph{symmetric} $t$ deletions and $s$ adjacent transpositions with low redundancy.
     \item Construct codes that are capable of correcting $t$ deletions/insertions + $k$ substitutions + $s$ adjacent transpositions.
     \item Construct codes for the Damerau--Levenshtein distance where the number of errors grows with the length of the code $n$ and is not restricted to constants $t$ and $s$.
 \end{itemize}

\bibliographystyle{IEEEtran}
\bibliography{references}

\begin{appendices}

\section{Proof of Lemma \ref{pro:upperbound}}\label{app:proof_uppersize}
In order to complete this proof, we require a few results concerning the approximate number of $\mathbf{0}^{i}1$ blocks in binary strings of length $n$ as $n\rightarrow \infty$, which appears as Claim~\ref{clm:typicalstring}. 

\begin{claim}[Lemma 3,~\cite{kovavcevic2019runlength}]\label{clm:typicalstring}
     In the asymptotic regime $n\rightarrow \infty$, binary strings of length $n$ contain $\sim$ $\lambda_{i}n$ blocks of $\mathbf{0}^{i}1$, where $\lambda_{i}=2^{-(i+2)}$.
\end{claim}

We now proceed to the proof of Lemma~\ref{pro:upperbound}.
\vspace{0.5em}

\noindent \textit{Proof of Lemma~\ref{pro:upperbound}:}
We proceed with the sphere-packing method to obtain an asymptotic upper bound on the code size. Let $\cC\subseteq\Sigma_2^n$ be a code that is capable of correcting $t$ $0$-deletions and $s$ adjacent transpositions.

For a codeword $\bx\in\cC$ with weight $w=\wt(\bx)$, let $\Gamma_{\neq 0}(_{\neq 0}\Gamma)$ denote the number of $1$s followed (preceded) by a 0-run with a length not equal to 0. In the presence of $s$ adjacent transpositions, we observe that a left-shift of 0 ($10\rightarrow 01$) occurs at the $1$ followed by a 0-run with a length not equal to 0, and a right-shift of 0 ($01\rightarrow 10$) occurs at the $1$ preceded by a 0-run with a length not equal to 0. Consequently, the number of strings obtained from $\bx$ through $s^{-}$ left-shifts of 0 ($10\rightarrow 01$) and $s-s^{-}$ right-shifts of 0 ($01\rightarrow 10$) can be expressed as follows:
  \begin{equation}\label{eq:size_trans}
      \sum_{s^{-}=0}^{s}{\Gamma_{\neq 0}\choose s^{-}}{_{\neq 0}\Gamma\choose s-s^{-}}.
  \end{equation}

  Similarly, we can determine that the number of $0$-runs is given by $r_0=\Gamma_{\neq 0}$. Therefore, the minimum number of strings obtained from $\bx$ through $t$ $0$-deletions is at least~\cite{levenshtein1966binary,levenshtein1965binary}:
  \begin{equation}\label{eq:size_del}
      {r_0-t+1 \choose t}={\Gamma_{\neq 0}-t+1\choose t}.
  \end{equation}

  Combining \eqref{eq:size_trans} and \eqref{eq:size_del}, we can conclude that the number of strings obtained from $\bx$ through $t$ $0$-deletions and $s$ adjacent transpositions is at least
  \begin{equation}\label{eq:size_del&trans}
      \sum_{s^{-}=0}^{s}{\Gamma_{\neq 0}\choose s^{-}}{_{\neq 0}\Gamma\choose s-s^{-}}{\Gamma_{\neq 0}-t+1\choose t}
  \end{equation}

From Claim~\ref{clm:typicalstring}, we can deduce that $\Gamma_{\neq i}\sim_{\neq i}\Gamma\sim \lambda_{i}n$, where $\lambda_{i}=2^{-(i+2)}$. Consequently, we obtain $w\sim n/2$ and $\Gamma_{\neq 0}\sim_{\neq 0}\Gamma\sim n/4$. Therefore, \eqref{eq:size_del&trans} can be asymptotically approximated as follows, utilizing the fact that ${n \choose m}\sim \frac{n^m}{m!}$ as $n\rightarrow \infty$, for fixed $m$:
\begin{align}
    &\sim \sum_{s^{-}=0}^{s}\frac{(n/4)^{s-}}{(s^{-})!}\frac{(n/4)^{s-s^{-}}}{(s-s^{-})!}\frac{(n/4)^{t}}{t!},\nonumber\\
    &\sim 2^{s}\frac{(n/4)^{s}}{s!}\frac{(n/4)^{t}}{t!}.
\end{align}

  Since $\cC$ is capable of correcting $t$ $0$-deletions and $s$ adjacent transpositions, we have the following: 
  \begin{equation*}
      2^n\gtrsim |\cC|\cdot \frac{(n/2)^{s}}{s!}\cdot\frac{(n/4)^{t}}{t!}.
  \end{equation*}

  Thus, we can conclude that an asymptotic upper bound on the maximal size of binary codes capable of correcting $t$ $0$-deletions and $s$ adjacent transpositions is given by:
  \begin{equation}
      \cM_{t,s}(n)  \lesssim \frac{2^n}{n^{s+t}}\cdot s!\cdot t!\cdot 2^{s+2t}.
  \end{equation}

\end{appendices}

\end{document}